\documentclass[10pt,twocolumn,twoside]{IEEEtran}
\usepackage{times}
\usepackage{amsmath,epsfig}
\usepackage{amssymb}
\usepackage{amsfonts}
\usepackage{array,dcolumn}

\usepackage{algorithm, algorithmic}
\usepackage{subcaption}
\usepackage{psfrag}
\usepackage{balance}
\usepackage{cite}
\usepackage[all]{xy}
\usepackage{dsfont}
\usepackage{hyperref}
\usepackage{url}
\usepackage{xcolor}
\usepackage[nolist]{acronym}
\usepackage{float}
\usepackage{threeparttable}
\usepackage{cancel}
\usepackage[official]{eurosym}
\usepackage{graphicx}
\usepackage[utf8]{inputenc}
\usepackage[english]{babel}
\usepackage{amsthm}
\usepackage[subnum]{cases}
\usepackage[nocomma]{optidef}
\usepackage[normalem]{ulem}
\usepackage{enumitem}

\usepackage[utf8]{inputenc}
\usepackage{pgfplots}
\usepackage{booktabs, adjustbox}
\DeclareUnicodeCharacter{2212}{−}
\usepgfplotslibrary{groupplots,dateplot}
\usetikzlibrary{patterns,shapes.arrows}
\pgfplotsset{compat=newest}
\usetikzlibrary{shapes.geometric, arrows}

\usepackage{multirow}

\usepackage{filecontents, pgffor}
\usepackage{listings, xcolor}

\usepackage{xcolor, colortbl}
\usepackage{mathtools,amssymb,lipsum, nccmath}
\usepackage{cuted}
\usepackage[normalem]{ulem}

\newcommand{\rv}[1]{\textcolor{black}{#1}}
\newcommand{\rrv}[1]{\textcolor{black}{#1}}

\newtheorem{definition}{Definition}

\newtheorem{Remark}{Remark}

\newtheorem{theorem}{Theorem}

\usepackage{tikz}
\usepackage{pgfplots}
\pgfplotsset{compat=1.10}
\usepackage{gincltex}
\usepgfplotslibrary{fillbetween}
\usetikzlibrary{patterns}

\def\BibTeX{{\rm B\kern-.05em{\sc i\kern-.025em b}\kern-.08em
		T\kern-.1667em\lower.7ex\hbox{E}\kern-.125emX}}

\hypersetup{
	colorlinks = true,
	linkcolor = blue,
	anchorcolor = black,
	citecolor = red,
	filecolor = blue,
	urlcolor = blue
}

\begin{document}
	\title{Strategic Coalition for Data Pricing \\ in IoT Data Markets}
	\markboth{IEEE Journal Submission}%
	{}
	\author{Shashi Raj Pandey,%
		\thanks{Shashi Raj Pandey and Petar Popovski are with the Connectivity Section, Department of Electronic Systems, Aalborg University, Denmark. Email: \{srp, petarp\}@es.aau.dk.}
		\textit{IEEE Member}, %
		Pierre Pinson,%
		\thanks{Pierre Pinson is with Dyson School of Design Engineering, Imperial College London, UK. Email: \{p.pinson\}@imperial.ac.uk. He is also affiliated to the Technical University of Denmark, Department of Technology, Management and Economics.} \textit{IEEE Fellow}, %
		and Petar Popovski, \textit{IEEE Fellow}
		\thanks{This work was supported by the Villum Investigator Grant “WATER” from the Velux Foundation, Denmark.}}
	
	\maketitle 
	\begin{abstract}
		\rv{This paper establishes a market for trading Internet of Things (IoT) data that is used to train machine learning models. The data, either raw or processed, is supplied to the market platform through a network, and the price of such data is controlled based on the value it brings to the machine learning model under the adversity of the correlation property of data. Eventually, a simplified distributed solution for a data trading mechanism is derived that improves the mutual benefit of devices and the market. Our key proposal is an efficient algorithm for data markets that jointly addresses the challenges of availability and heterogeneity in participation, as well as the transfer of trust and the economic value of data exchange in IoT networks. The proposed approach establishes the data market by reinforcing collaboration opportunities between devices with correlated data to limit information leakage. Therein, we develop a network-wide optimization problem that maximizes the social value of coalition among the IoT devices of similar data types; at the same time, it minimizes the cost due to network externalities, i.e., the impact of information leakage due to data correlation, as well as the opportunity costs.} Finally, we reveal the structure of the formulated problem as a distributed coalition game and solve it following the simplified split-and-merge algorithm. Simulation results show the efficacy of our proposed mechanism design toward a trusted IoT data market, with up to $32.72\%$ gain in the average payoff for each seller.
	\end{abstract}
	\IEEEpeerreviewmaketitle
	
	\section{Introduction}
	\subsection{Context and Motivation}
	The massive volume of Internet of Things (IoT) devices and services lead to an exponential growth of IoT data \cite{data794}. Various networked cyber-physical systems (CPSs) are accumulating and processing data at a large scale, often contributing to training some learning model or carrying out an inference. For instance, massively distributed data, when integrated with Machine Learning (ML) tools, stimulate both \textit{real time} and \textit{non-real time} decision-making services that create a value of data in the IoT networks~\cite{popovski2021internet}. This brings the question of economic opportunities in IoT data markets, where data and its value to the services can be traded or exchanged. It is thus relevant to study the IoT data markets in terms of mechanisms for attaining the desired economic properties in offering learning services, such as prediction, detection, classification, forecasting, and similar. Furthermore, it is necessary to investigate strategies involved in the execution of such distributed cooperation amongst devices having data of value for IoT data markets. 
	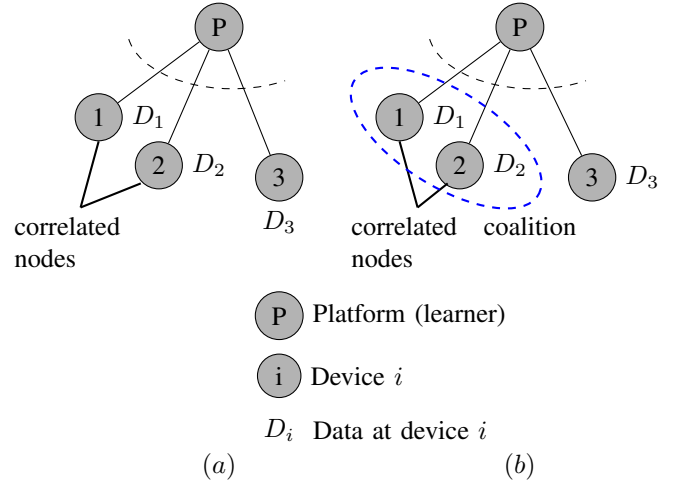
\begin{figure}
    \centering
\begin{tikzpicture}[scale=0.8]
  \node (p1)      at ( 0,0) [shape=circle,draw, fill=black!30] {P};
  \node (d1)    at ( -2,-1.5) [shape=circle,draw, fill=black!30,  label=right:$D_1$] {1};
  \node (d2)      at ( -1,-2.3) [shape=circle,draw, fill=black!30, label=right:$D_2$] {2};
  \node (d3) at ( 1,-2.5) [shape=circle,draw, fill=black!30, label=below:$D_3$] {3};

   \draw [-] (p1) to (d1);
   \draw [-] (p1) to (d2);
    \draw [-] (p1) to (d3);

    \draw [dashed] (-1.5,-0.2) arc [start angle=190, end angle=300, x radius=1.75cm, y radius=1cm];
    
    \draw[thick] (-2, -1.9) -- (-2.3, -3) node[anchor= west, text width=5em] {};
    \draw[thick] (-1.3, -2.6) -- (-2.3, -3) node[anchor= north, text width=5em] {correlated nodes};

  \node (p2)      at ( 5,0) [shape=circle,draw, fill=black!30] {P};
  \node (2d1)    at ( 3,-1.5) [shape=circle,draw, fill=black!30,  label=right:$D_1$] {1};
  \node (2d2)      at ( 4,-2.3) [shape=circle,draw, fill=black!30, label=right:$D_2$] {2};
  \node (2d3) at ( 6.2,-2.5) [shape=circle,draw, fill=black!30, label=right:$D_3$] {3};

   \draw [-] (p2) to (2d1);
   \draw [-] (p2) to (2d2);
    \draw [-] (p2) to (2d3);

    \draw [dashed] (3.5,-0.2) arc [start angle=190, end angle=300, x radius=1.75cm, y radius=1cm];
    
    \draw[thick] (3, -1.9) -- (3.3, -3) node[anchor= west, text width=5em] {};
    \draw[thick] (3.8, -2.6) -- (3.3, -3) node[anchor= north, text width=5em] {correlated nodes};
 
     \draw [dashed] (5.3,-2.8) [ rotate=-32, blue, thick] arc [start angle=0, end angle=360, x radius=1.8cm, y radius=0.8cm];
 
     \draw[thick] (5.5, -3)  node[anchor= north, text width=5em] {coalition};

    \node (p)      at (1,-4.8) [shape=circle,draw, fill=black!30, label=right:$\textrm{Platform (learner)}$] {P};
    
    \node (p)      at ( 1,-5.8) [shape=circle,draw, fill=black!30, label=right:$\textrm{Device} \ i $] {i};
    
    \node (p)      at ( 1,-6.7) [label=right:$\textrm{Data at device}\ i$] {$D_i$};   
        
    \node (label a)      at ( 0,-7.3) {$(a)$};  
    \node (label b)      at ( 5,-7.3) {$(b)$};  
  
\end{tikzpicture}
    \caption{\rv{Devices with correlated data and coalition formation when interacting with the platform (learner). The bold grey lines between devices and the platform indicate interaction interface.}}
    \label{fig:coalition}
\end{figure}
	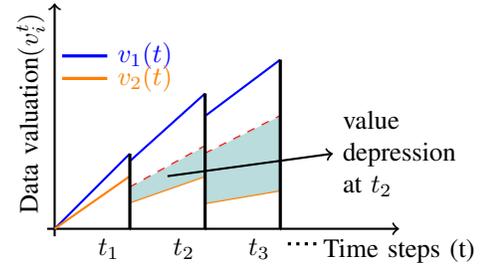
\begin{figure}[t!]
\centering
\begin{tikzpicture}
\draw[thick,->] (0, -0.1) -- (0, 3) node[anchor= south east, rotate = 90] {$\textrm{Data valuation} (v_i^t)$};
\draw[thick,->] (-0.1, 0) -- (4.58, 0) node[anchor=north] {$\textrm{Time steps (t)}$};
\draw[scale=1, domain=0:1, smooth, variable=\x, blue, thick] plot ({\x}, {\x});
\draw[scale=1, domain=1:2, smooth, variable=\x, blue, thick] plot ({\x}, {0.9*\x});
\draw[scale=1, domain=2:3, smooth, variable=\x, blue ,thick]  plot ({\x}, {0.75*\x});
\draw[scale=1, domain=0:1, smooth, variable=\x, orange ,thick]  plot ({\x}, {0.7*\x});
\draw[scale=1, domain=1:2, dashed, variable=\x, red ,thick, name path=a]  plot ({\x}, {0.55*\x});
\draw[scale=1, domain=2:3, dashed, variable=\x, red ,thick,, name path=a1]  plot ({\x}, {0.5*\x});    
\draw[scale=1, domain=1:2, smooth, variable=\x, orange ,thick,  name path=b]  plot ({\x}, {0.35*\x});
\draw[scale=1, domain=2:3, smooth, variable=\x, orange ,thick,, name path=b1]  plot ({\x}, {0.17*\x});
 
\tikzfillbetween[of=a and b]{teal!30,opacity=0.9};
\tikzfillbetween[of=a1 and b1]{teal!30,opacity=0.9};
\draw[thick,->] (1.5, 0.7) -- (3.7, 1) node[anchor= west, text width=5em] {value depression at $t_2$};

\draw[black,very thick, dotted] (3.1,-0.2) -- (3.5,-0.2);
\draw[black,very thick] (1,0) -- (1,1);
\draw[black,very thick] (1,0)  -- (1,0) node[anchor= north east] {$t_1$};
\draw[black,very thick] (2,0) -- (2,1.8);
\draw[black,very thick] (2,0)  -- (2,0) node[anchor= north east] {$t_2$};
\draw[black,very thick] (3,0) -- (3,2.25);
\draw[black,very thick] (3,0)  -- (3,0) node[anchor= north east] {$t_3$};
\draw[blue, very thick] (0.1,2.3) -- (0.7,2.3) node[anchor= west] {$v_1(t)$};
\draw[orange, very thick] (0.1,2) -- (0.7,2) node[anchor= west] {$v_2(t)$};
\end{tikzpicture}
\caption{Value depression due to leakage of correlated information in a two seller $i\in\{1,2\}$ one buyer scenario in a linear pricing scheme.}
\label{fig:value_depress}
\end{figure}
	\rrv{The two fundamental aspects of an IoT data market are \cite{acemoglu2022too, agarwal2019marketplace, ding2022optimal, mao2019pricing, niyato2016market}: (1) offered \emph{pricing}\footnote{Pricing indicates monetary reward or incentives of any form in general, such as discount vouchers.}, and (2) device participation in the data trading process.} An IoT device can be stimulated by the network to participate and share data of value. \rv{The stimulation is achieved by pricing signals that compensate the IoT device based on the data valuation and the cost of data privacy, with additional computational and communication costs. Most of the work in the IoT data market, therefore, focuses on economic literature of IoT networks \cite{ghosh2018pricing, mao2019pricing}: designing proper pricing mechanisms that factor in these costs with the device's willingness-to-sell data for several IoT-based services while ensuring revenue/utility maximization objectives, and data trading protocols, as in \cite{ding2022optimal, nguyen2021marketplace, nguyen2021modeling}. The authors in \cite{oh2020competitive} characterize such interactions for data exchange with a competitive data trading model. In \cite{pandey2020crowdsourcing}, the authors design an incentive mechanism that aligns strategies of distributed IoT nodes on participation to train a high-quality ML model via sharing data as informative local learning parameters. Likewise, several other works consider data valuation and pricing to realize IoT data markets offering heterogeneous services, such as ML training, statistical analysis, and forecasting with the distributed data \cite{pandey2022fedtoken, maddikunta2022incentive, pinson2022regression}.}
	
	\rv{As discussed, devices respond \emph{effectively} only at the right pricing signals that resolve their privacy concerns. In such terms, \emph{perfect data privacy} can be seen as a case of having data with \emph{infinite pricing}. Therein, a mild data privacy preference allows the market to adjust its pricing offers for the devices to join in the data trading, revealing a trade-off between data privacy and pricing. However, given the appropriate pricing signal for participation, the statistical properties of traded data over IoT networks raise fundamental challenges on the scope/impact of exploiting data trading mechanisms to realize the IoT data market. Particularly by adding a different dimension to consider data privacy, where one's data can be inferred (or learned) by abusing other's data. This downplays the pricing offers and potential value of data during participation; hence, demanding a meticulous investigation of data privacy, pricing and participation to realize IoT data markets, which has been overlooked in the recent literature \cite{ghosh2018pricing, mao2019pricing, ding2022optimal, pandey2020crowdsourcing, pandey2022fedtoken, maddikunta2022incentive, marjani2017big,saputra2019energy,liu2020privacy,vepakomma2021private}.}
	
	The seminal work \cite{acemoglu2022too} models the data market by relating the correlation among devices' data to price depression. To illustrate the main ideas of~\cite{acemoglu2022too} in a data streaming setting, consider an IoT data market without any data privacy consideration, as in Fig.~\ref{fig:coalition}(a), featuring devices with correlated data that interact and trade data with the platform\footnote{Note that in this case, the platform acts as a learner, such that we will use the terms learner and platform interchangeably.}. Assume that the device $\{1\}$ shares its dataset $\mathcal{D}_1$ with the platform, after which the device $\{2\}$ does the same. If the two datasets are correlated, and the market learns the total variation distance between $\mathcal{D}_1$ and $\mathcal{D}_2$, it can prioritize pricing for the earliest traded data and drop offered price rate for the latter. Consider Fig.~\ref{fig:value_depress}, which shows the evolution of data valuation over time $t\ge0$, assuming a linear pricing scheme. The market defines $v_i(t)$ as the valuation of the data of seller $i$ at time $t$.
	We observe that the valuation of data $v_2(t)$ for the seller $\{2\}$ drops in subsequent interactions with the market because device $\{1\}$ leaks information about the data of device $\{2\}$. 
	
	\rrv{To elaborate and formalize this problem, take the IoT data market scenario in Fig.~\ref{fig:datamarket}. The platform in Fig.~\ref{fig:datamarket} acts as an interface through which the devices, i.e., the buyers and sellers, interact in the IoT data market. Particularly, the platform is used for data storage and/or data and value exchange. A number of IoT devices connected to a platform collaboratively train a learning model and create value in a privacy-preserving manner; e.g., this could be a predictor in the Federated Learning (FL) setting \cite{saputra2019energy,hard2018federated,liu2020privacy, mcmahan2017communication}.} In this regard, the authors in \cite{agarwal2019marketplace} discuss a marketplace for data where a robust Shapley technique is developed to capture replicable properties of exchanged data and ways to capitalize the value when sharing them. In general, such a market offers incentives to the devices in a way that (i) stimulates their participation \cite{pandey2020crowdsourcing, le2021incentive, jiao2020toward}, and (ii) strikes a balance between the data privacy concerns, the trustworthiness of the data market, and the cost of data trading \cite{ali2020voluntary, acemoglu2022too}. \rrv{As explained before, and in \cite{acemoglu2022too}, the market may also leverage the correlated information or \emph{information leakage}, and other statistical properties of data between sellers \cite{gupta1999modeling, ichihashi2021economics, rubin1993statistical, agarwal2019marketplace}, to steer the pricing signals for self-benefit unilaterally. This leads to uncontrolled competition in data sharing, particularly due to data rivalry. As the market exploits more data, it causes devices to drop their participation out of mistrust or negligible pricing. This network externality creates a \textit{loop of mistrust}\footnote{\rrv{A scenario where the agents unknowingly behave competitively under the influence of the pricing signal due to the potential information leakage caused by data similarity.}}, by which the platform can manipulate the data market, causing and affecting device participation in the data trading process.} \rv{However, as discussed, an integrated view on data pricing and the participation for data/value exchange under the adversity of network externalities have not been studied in the related works yet. Motivated by this observation, the key contribution of this work is a method by which the devices with data privacy concerns can challenge the market to limit price depression: forming a coalition within devices with correlated data, as shown in Fig. \ref{fig:coalition}(b), and bargaining as a group instead of individually during participation. In the following, we systematically elaborate on the novelties and contributions of our work.}
	
	\begin{figure}[t!]
		\centering
		\includegraphics[width=\linewidth]{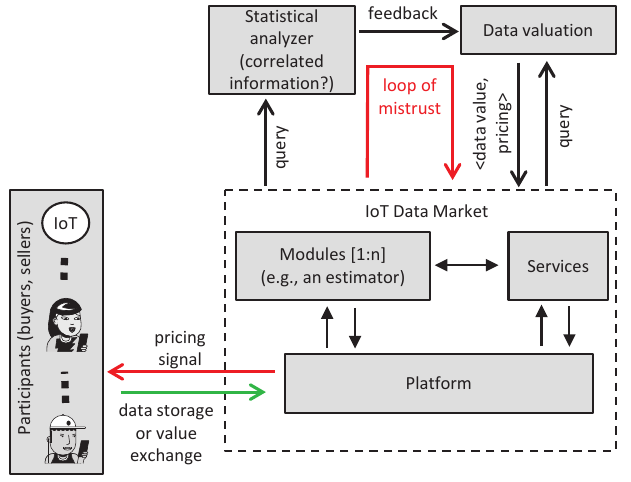}
		\caption{A schematic framework: the loop of mistrust in the IoT data market due to information leakage.}
		\label{fig:datamarket}
	\end{figure} 
	\subsection{Challenges and Contributions}\label{subsec:challenges}
	\rv{Our solution to counter the loop of mistrust consists of two steps with the following motivation: (1) form a coalition among a group of devices to tighten the information leakage within the group; (2) challenge the platform to execute the data trading mechanism in a trusted setting. Overall, we show this brings value in collaboration with improved pricing offers. Coalition formation limits data rivalry amongst sellers, lowers the impact of information leakage due to uncontrolled competitive data trading on pricing, and fosters \emph{availability} of devices to establish data markets. However, forming a coalition to realize a data market is not straightforward since the devices need to: (i) learn correlated statistical properties of data of the other devices, and without revealing it through the market, (ii) characterize and formalize relevant utility models that identify conditions for coalition formation and price determination amongst devices within the coalition, and (iii) handle time-complexity and efficiency of coalition formation at scale.}
	
	Based on the above discussion, we have identified the challenges expressed through the following questions: 
	\begin{itemize}
		\item \textbf{Q1}: \emph{How do IoT devices protect correlated information of their data without allowing the learner to manipulate pricing in the IoT data market?}
		\item \textbf{Q2}: \emph{How can the platform infer when the data from devices have equal marginal values or valuations?}
		\item \textbf{Q3}: \emph{What is the impact of device availability in the data trading process?}
	\end{itemize}
	
	\rv{Addressing \textbf{Q1} means preventing the market from identifying a possible correlation between different data types of devices and, further, monopolising the pricing and data trading strategies of devices. To eliminate this, we devise distributed coalitions of devices with similar data types, which enforces the platform to derive the marginal contribution of the coalition, or simply the coalition value, instead of the individual interactions. This reveals the game-theoretic interactions between the data holders (i.e., devices) and the IoT data market. Addressing \textbf{Q2} positions us to develop reasonable utility models for the IoT market that equally benefit the platform without hurting participation of devices in coalition due to information leakage and unreliable connectivity in the IoT networks\footnote{In this work, we realize unreliable connectivity in terms of participation. Particularly, connectivity is considered an uncontrollable factor in the IoT network, wherein we reflect the unreliable connectivity with the availability of market players, i.e., the devices for data trading.}. Then, we show this eventually leads to the formation of different coalition structures that balances individual payoffs and the stability of the coalition.  Addressing \textbf{Q3} positions us to evaluate the value of participation and sets us to develop an integrated framework that jointly incorporates \textbf{Q1} and \textbf{Q2} in the IoT data market design.} 
	
	\rv{As a main contribution, we develop a \emph{novel cooperation protocol}, termed multi-agent joint policy (MAJP), in an IoT data market that enables devices to maximize their value of participation in the coalition. This avoids the situation in which devices compete with each other, thereby possibly leading to leakage of statistical information about the data to be traded on the market. We first highlight the \textit{tension between data sellers and the data market over offered pricing and issues due to information leakage}, and then offer a distributed solution to overcome the presented challenges. In doing so, we first explore the characteristics of the formulated MAJP problem, which is intractable due to binary constraints and the coupling of variables. Therein, we devise a game-theoretic mechanism that offers a simplified, distributed solution to the MAJP problem where the proposed approach reinforces data sellers into collaboration for data trading with the objective of minimizing information leakage in a distrusted IoT data market. The proposed mechanism of distributed coalition games captures the properties of information leakage, the value of collaboration and the opportunity costs during the coalition formation. We show the developed solution is of low-complexity, with convergence guarantees. To that end, we derive stability conditions of the coalition game following the developed distributed coalition formation mechanism based on the merge-and-split algorithm \cite{apt2009generic} to realize a trusted IoT data market. Finally, through a sequence of numerical evaluations, we show the efficacy of our proposed mechanism.}
	
	The rest of the paper is organized as follows. Section \ref{sec:preliminaries} provides preliminaries, introduces a truthful IoT data market model and defines device data type. Section \ref{sec:problem formulation} develops the underlying utility models and presents the problem formulation as an optimization problem to obtain a multi-agent joint policy (MAJP) in a distributed coalition setting. Section \ref{sec: market design} develops the market design and proposes a coalition game solution to the MAJP problem with complexity analysis. Section \ref{sec:simulations} provides the performance evaluation of the proposed approach and shows comparative analysis with the competitive baselines. Finally, Section \ref{sec:conclusion} concludes this work.
	
	\section{Preliminaries and Problem Setting}\label{sec:preliminaries}
	\subsection{Network Setup for IoT Data Market}
	\begin{table}[t!]
		\centering
		\caption{Notation.}
		\label{tab:table1}
		\begin{tabular}{ll}
			\hline
			Notation & Definition\\
			\hline
			$\mathcal{M}, \mathcal{D}$ & Set of IoT devices and total data samples, respectively \\
			$\mathcal{D}_m$ & Set of data samples at device $m$ \\
			$\tilde{J}$ & Potential function of the global loss $J$ with \\& supporting mini-batch $z \in \mathcal{D}_m$ of local data samples \\  
			$\Phi_m$ &  Data type of device $m$\\           
			$\mathcal{S}$ & Coalition space \\             
			$\mathcal{D}_\mathcal{S}$ & Collective data samples with a set of devices in $\mathcal{S}$ \\        
			$\theta_m(z, \mathcal{D}_m)$ &  Average marginal contribution of device $m$ with data samples $\mathcal{D}_m$\\              
			$\xi_m(\theta_m)$ &  Composite value mapping as privacy preference profile of device $m$\\            
			$\Pi$ &  Set of coalition\\ 
			$\Omega$ &  Outcome space to evaluate performance on participation\\
			$\mathcal{A}$ & Action space \\              
			$a_m(t)$  & Participation variable for device $m$\\
			$p_m(t)$  & Offered pricing for device $m$\\            
			$\textbf{a}$ & Participation vector \\ 
			$\textbf{p}$  & Pricing vector \\
			$f_{\mathcal{S}}$ & Utility function defined over the coalition $\mathcal{S}$\\
			$v(\mathcal{S, \mathcal{A}})$ & Coalition value \\
			$\rho_\mathcal{S}$ & Average data dissimilarity in the coalition $\mathcal{S}$\\
			$c(\mathcal{S})$ & Cost of coalition \\
			$c(\Delta\phi_i)$ & Opportunity cost in coalition $i \in \mathcal{S}$ \\   
			$u_m(\mathcal{S})$ & Preference function of device $m$\\
			\hline
		\end{tabular}
	\end{table}
	
	\rrv{A typical IoT data market considers the interaction between the buyer and the seller owning IoT devices to trade the data or services (e.g., training learning models) with pricing signals \cite{niyato2016market, mao2019pricing, li2021capitalize}.} Consider a network with a finite set of IoT devices $\mathcal{M}$ as $|\mathcal{M}| = M$ training a global learner (e.g., a predictor\footnote{The proposed model is generic in the sense that it works well for training learning models, such as in FL, or an estimator minimizing the mean square error (MSE).}). The learner acts as an intermediary (for simplicity, we consider it as the buyer, or equivalently, the platform) purchasing data from the distributed devices (sellers); thus, forming a marketplace where strategic data sellers get incentives for their contribution to improving the model at the global learner. Indeed, to explore the strategic interaction between the devices and the learner, the following remark is useful.
	\begin{Remark}\label{rmk:game}
		In a game-theoretic setting, the said set of $\mathcal{M}$ participating devices are often referred to as ``agents". These finite set of agents hold explanatory data samples for specific learning tasks and aim to exchange it with the learner, fully or in a privacy-preserving manner, e.g., in FL \cite{mcmahan2017communication}, for training learning models of interest to the learner. Then, any rational agent is willing to participate in data trading, given offered pricing compensates their cost of participation.
	\end{Remark}
	Each device $m\in\mathcal{M}$ stores the data samples at time $t\in \mathcal{T} = \{1,2, \ldots, T-1\}$, defined as the local data set $\mathcal{D}_m (t)$ of size $D_m (t)$. Note that, in a typical distributed learning mechanism under the synchronous settings, the observation time $t$ is a single round of global interaction between the platform and the devices \cite{mcmahan2017communication, pandey2020crowdsourcing}. Then, the collective data sample size at time $t$ across the network is $D (t)=\sum_{m=1}^{M}D_{m}(t)$. In a supervised learning setting, $\mathcal{D}_m (t)$ is a collection set of data samples at device $m$ defined as $\{x_i, y_i\}_{i=1}^{D_m}$ with $x_i \in \mathbb{R}^d$ corresponding label $y_i \in \mathbb{R}$. The data samples are informative about the learning model; hence, bring value to the learner in terms of their contribution to improving the learning performance. We refer to it as \emph{data type}. Following this intuition, we associate the type of data samples available at the devices as a realization of random variable $\Phi$ (explained in Definition~\ref{def:type}). This setting can be extended to a more generalized form where each device shares a stream of data samples with features accounting for the time instance, such as in the time series prediction. 
	
	Consider $n \in \mathcal{N}(t)$ data samples available in the network for trading such that $D_m(t) \le |\mathcal{N}(t)|$. Then, the goal of a supervised learner is to learn a single model defined in Definition~\ref{def:learner}.
	\begin{definition}\label{def:learner}
		We define a supervised learner interested in minimizing the empirical risk with respect to parameter $w\in\mathbb{R}^d$ on all distributed data samples $D_m(t)$ as the finite-sum objective of the form
		\begin{equation}
			\underset{w \in \mathbb{R}^d}{\min}J(w, t) \ \  \textrm{where} \ \ J(w,t) := \sum\nolimits_{m = 1}^M \frac{D_m(t)}{D(t)}J_m (w, t).
			\label{eq:learning_problem}
		\end{equation}
	\end{definition}
	Then, the data market particularly looks at the contribution of each device $m$ in solving \eqref{eq:learning_problem}, which is expressed as the empirical risk with respect to the improvement in $w\in\mathbb{R}^d$ on their local data set $\mathcal{D}_m(t)$ as
	\begin{equation}
		J_{m}(w, t) := \frac{1}{D_m(t)} \sum \nolimits_{i = 1}^{D_m(t)}f_i(w).
		\label{eq:localloss}
	\end{equation}
	For simplicity and without loss of generality, we make a common assumption: $f_i(w)$ is $(1/\gamma)$-smooth function or a $L$-Lipschitz continuous function (cf. \cite{konevcny2016federated}); hence, ensuring convergence and stability of the solution.
	We note the network topology is not restrictive towards changes, i.e., the devices can perform data trading with each other via an arbitrator (a central learner) because they are connected through a network. Therefore, for the performance analysis of the proposed approach, later on, we provide a scenario-based statistical analysis.
	
	\subsection{Data Type and IoT Data Market Model}
	We make a common assumption that the market is interested in data exchange and, therefore, stimulates the devices with pricing signals based on the value of the traded data in improving learning performance. In our setting, this translates to first finding the type of data each device has and its possible influence on model training. Then, for the offered pricing $p_m >0$, every rational device $m$ determines its strategy for participation $a_m \in \{0,1\}$ in the data trading process so as to maximize their individual benefits; such strategies are captured in terms of the defined utility function. For this, we first define the type of data samples $\mathcal{D}_m$ of device $m$ by a random variable $\Phi_m$. To this end, we have the following definition. 
	\rv{\begin{definition}\label{def:type}
			We define the data type  $\Phi_m, \forall m \in \mathcal{M}$ as a composite measurement obtained following the privacy preference profile $\xi_m \in [0,1]$ of the devices to share their available data, fully or partially, and the average marginal contribution value of supporting mini-batch of data points $z \in \mathcal{D}_m$ defined as $\theta_m(z, \mathcal{D}_m)$ that brings to the learner.
	\end{definition}}
	\begin{Remark}\label{def:remark}
		Here we remark that sharing data partially means the exchange of local learning models, as in FL training, where instead of raw data, only data related to the local model parameters is shared in the data market; hence, the FL approach of data trading is considered `reasonably' privacy-preserving. In either case, i.e., for any arbitrary privacy preference, the exchange of devices' data brings an exogenous value in the data market, but only at the right price offered for their data type.
	\end{Remark}
	Following Definition \ref{def:type} and Remark \ref{def:remark}, we formally define the composite mapping $\xi_m(\theta_m(z, \mathcal{D}_m))$, where $\theta_m(z, \mathcal{D}_m)$ can be evaluated following a modified distributed Shapley \cite{jia2019towards} value for a known potential function $\tilde{J}$ of the global loss\footnote{Potential function reflects the performance metric in terms of learner's model accuracy, which is supplied to the market at the beginning of the economic interaction.} (as defined in \eqref{eq:learning_problem}) such that 
	\begin{equation}
		\theta_m(z, \mathcal{D}_m) = \underset{\underset{\tilde{D} \sim \mathcal{D}_m^{i-1}}{i\sim[B]}}{\mathbb{E}}\bigg[\tilde{J}(\tilde{D}\cup \{z\}) -\tilde{J}(\tilde{D})\bigg],
		\label{eq:type}
	\end{equation}
	where, respectively, $B$ is the mini-batch size, $\tilde{D}$ is the i.i.d. samples drawn from the available data samples $\mathcal{D}_m$ supported on $\mathcal{Z}$ with mini-batch of data point $z \in \mathcal{Z}$, and the potential function $\tilde{J}: \mathcal{Z} \rightarrow [0,1]$ defined by the output $0 \le \epsilon \le 1$ such that $|\nabla \tilde{J}(w^{(t)})|  \le \epsilon |\nabla \tilde{J}(w^{(t-1)})|$.
	
	More precisely, following \eqref{eq:type}, the mapping $\xi_m(\theta_m(z, \mathcal{D}_m)) \in [0,1]$ quantifies data type as the expected value of data and the device's preference profile, i.e., the willingness to trade data with the learner, to offer that value in the data market. For simplicity, we use shorthand $\xi_m $ for $\xi_m(\theta_m(z, \mathcal{D}_m))$, with $\xi_m = 0$ when the device $m$ reserves no privacy concern  on the shared data. \rv{Therein, devices undergo strategic participation at the right pricing. In practice, we observe heterogeneity in $\xi_m$, which is an important metric that captures the function of individual preference on sharing data (i.e., data privacy), and thereof, each device may not reveal its true data type or perform optimal local computation, as expected by the market, for the offered pricing scheme to participate in the data trading process. Hence, the learner face consequences of the partial knowledge in the state of information exchanged\footnote{This is often termed as \emph{information asymmetry}.}} in a setting where payments for traded data are provided after collecting them. We settle the aforementioned analysis with the formalization of an efficient trading mechanism in the proposed market model as below.
	\begin{definition}\label{def:mechanism}
		The proposed data trading mechanism is a tuple $(\Pi, \Omega, \mathbf{p}, \mathbf{a})$, where $\Pi$ is the coalition set following data types, $\Omega$ is the outcome space capturing the final learning performance, with $\Omega: \Phi \times \mathbf{a} \rightarrow [0,1]$, $\mathbf{p} = (p_m(\Phi_m))_{m \in \mathcal{S}_{\Pi}}$ is the pricing vectors defined for coalition $\mathcal{S}_{\Pi}$, with $a_{m \in \mathcal{S}_{\Pi}}=1$, and $\mathbf{a}= (a_m(\Phi_m))_{m \in \mathcal{S}_{\Pi}}$ is the vector of participation to tighten the information leakage due to data correlation.
	\end{definition}
	Definition~\ref{def:mechanism} hints at the underlying game-theoretic interaction between the learner and the devices for the mechanism design, summarized as the following. The learner (i) evaluates the received data (including the device's importance value towards privacy) and aims at (ii) quantifying the type of device's data so as to lower the offered pricing. Particularly, the traded data is evaluated for its contribution to improving the performance of the learning model, i.e., with $\Omega: \Phi \times \mathbf{a} \rightarrow [0,1]$. Whereas, the devices $m \in \mathcal{M}$ with correlated data samples form a coalition $m \in \mathcal{S}_{\Pi}$ to challenge the learner in hiding their own data type $\Phi_m$, or adopt sharing data in bundles to mitigate the information leakage and price depression. Following Definition~\ref{def:mechanism}, the mechanism aims to foster improved participation in training learning models while addressing the impacts of data correlation on the offered pricing.
	
	With these preliminaries, next, we formally start to tackle the research problems \textbf{Q1}, \textbf{Q2}, and \textbf{Q3}, raised in Section \ref{subsec:challenges}, with the considered simple setting. In the following, we present an overview of the problem formulation about data trading in the IoT data market and formalize the data valuation procedure as per the data properties, resulting in specific utility models.
	
	\section{Problem Formulation}\label{sec:problem formulation}
	\subsection{A basic setup}
	We revisit Definition \ref{def:type} and make an assumption that the vector of random variables $\Phi = [\Phi_1, \Phi_2, \ldots, \Phi_m]$ follows a joint normal distribution $\mathcal{N}(\mu_\Phi,\Sigma )$, where $\Sigma\in\mathbb{R}^{m\times m}$ is the covariance matrix. This setup provides convenience in further analysis; we simply assume this to reflect the presence of devices with the correlated data types. However, the developed framework is not limited to this assumption, as in the case of otherwise, the problem eventually boils down to the deconstruction of the data type, and our approach follows. Consider $a_m(t)$ as a binary decision variable for device $m$ to join the data market such that
	\begin{equation}
		a_{m}(t) =
		\begin{cases}
			1, \; \; \text{ if device $m$ joins the market at time $t$},\\
			0, \; \; \text{otherwise}.
		\end{cases}
	\end{equation}
	Then, in every round of interaction with the learner for the offered pricing $p_m(t), \forall m$, the interested device (if in the agreement to participate) trade their data as the mixture of their data type and the learning parameters such that $S_m =f(\xi_m)+N_m$, where $N_m\sim\mathcal{N}(0,1)$ is the Gaussian noise. In this regard, as shown in \cite{acemoglu2022too}, the learner can have an estimate of $\xi_m$ with the traded data $D_m$ with a solution to minimization of the estimation error of the data type. In doing so, the learner can employ both convex/non-convex loss function in \eqref{eq:type} that defines the data type of a device. This means, the learner can efficiently reconstruct the mapping function $\hat{\phi}=<g(\xi_m)|_{m \in \mathcal{M}}>$ to derive $\xi_m, \forall m$ precisely by solving the squared-error minimization problem as
	\begin{equation}
		\underset{g(\xi_m)}{\arg\min}\; \mathbb{E}\bigg[(\phi_m - g(\xi_m|D_m, a_m(t), p_m(t)))^2\bigg], \forall m \in \mathcal{M}.
		\label{eq:mse}
	\end{equation}

\begin{figure}
    \centering
    \begin{tikzpicture}
    
    \definecolor{darkgray176}{RGB}{176,176,176}
    \definecolor{lightgray204}{RGB}{204,204,204}
    \definecolor{orange}{RGB}{255,165,0}
    
    \begin{axis}[
    legend cell align={left},
    legend style={
      fill opacity=0.8,
      draw opacity=1,
      text opacity=1,
      at={(0.97,0.03)},
      anchor=south east,
      draw=lightgray204
    },
    tick align=outside,
    tick pos=left,
    x grid style={darkgray176},
    xlabel={Volume of correlated data samples},
    xmajorgrids,
    xmin=0.1, xmax=4,
    xtick style={color=black},
    xtick={1,2,3,4},
    xticklabels={10,20,30,40},
    y grid style={darkgray176},
    ylabel={Scaled valuation function},
    ymajorgrids,
    ymin=1.04067693290093, ymax=1.52187252221815,
    ytick style={color=black}
    ]
    \addplot [very thick, orange, dashed, mark=square*, mark size=2, mark options={solid}]
    table {%
    0 1.45939946174622
    1 1.49450528621674
    2 1.49925637245178
    3 1.49989938735962
    4 1.49998641014099
    };
    \addlegendentry{$A_0 = 0.2, \textrm{w/o} \; n_0$}
    \addplot [very thick, orange, mark=triangle*, mark size=2, mark options={solid}]
    table {%
    0 1.42142856121063
    1 1.43716371059418
    2 1.48202574253082
    3 1.49333846569061
    4 1.499431848526
    };
    \addlegendentry{$A_0 = 0.2, \textrm{w/} \; n_0$}
    \addplot [very thick, blue, dashed, mark=square*, mark size=2, mark options={solid}]
    table {%
    0 1.31729733943939
    1 1.47527384757996
    2 1.49665367603302
    3 1.499547123909
    4 1.49993872642517
    };
    \addlegendentry{$A_0 = 0.7, \textrm{w/o} \; n_0$}
    \addplot [very thick, blue, mark=triangle*, mark size=2, mark options={solid}]
    table {%
    0 1.06254947185516
    1 1.41216039657593
    2 1.46816420555115
    3 1.45881950855255
    4 1.49451267719269
    };
    \addlegendentry{$A_0 = 0.7, \textrm{w/} \; n_0$}
    \end{axis}
    \end{tikzpicture}
    \caption{\textbf{Case study} on variability in scaled valuation function at the learner in terms of model precision: two sellers one buyer scenario.}
    \label{fig:valuation_func}
\end{figure}
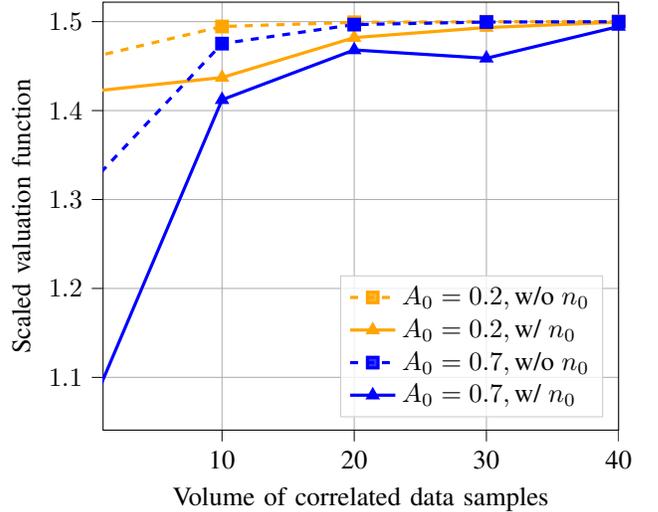
	In this regard, we outline the following three cases:
	\begin{itemize}
		\item \textbf{Case I:} When $\xi_m = 0$, i.e., $f^{-1}(\phi_m)=0, \forall m$, the learner adopts the following economic properties.
		\begin{enumerate}[label=(\roman*)]
			\item \textbf{Monotonicity:} If we have $\mathcal{D}_i \subseteq \mathcal{D}_j$ for any pair of devices $i,j\in \mathcal{M}$, then following the standard assumption of monotonicity in the valuation on data defined as $v(\cdot)$, we have $v(D_i)\leq v(D_j)$. 
			\item \textbf{Additive:} If we have $\mathcal{D}=\mathcal{D}_i\cap \mathcal{D}_j$ for any pair of devices $i,j\in \mathcal{M}$, then we have $v(D)\le v(D_i)+v(D_j)$.
		\end{enumerate}
		Considering these two properties, if the learner already received data $\mathcal{D}_i$ such that $\mathcal{D}_i, \mathcal{D}_j \subseteq \mathcal{D}$, then it formalizes the valuation function for the data $\mathcal{D}_j$ as
		\begin{equation}
			v^{D_i}(D_j) = \gamma\cdot v(\mathcal{D}_i\cap \mathcal{D}_j) + v(\mathcal{D}_j\setminus \mathcal{D}_i|\mathcal{D}_i\cap \mathcal{D}_j), \forall i,j \in \mathcal{M}
		\end{equation}
		where $\gamma\ge1$ is a design parameter quantifying the effect of available data samples on the learner. The right-hand part characterizes the marginal contribution of the remaining data samples. In fact, the valuation can be explicitly defined in terms of its contribution to improving learning performance. As an example, we discuss the following case study.\\
		\noindent \textbf{Example 1:} Take a log-concave valuation function of common data samples, an example, $\mathcal{D}_i\cap \mathcal{D}_j$, as defined according to the experimental results in \cite{zhan2020learning}, resulting the learning precision (or accuracy) $\zeta$ such as $J(\zeta)$, where $\zeta = 1-A_0e^{-2|\mathcal{D}_i\cap \mathcal{D}_j|(1-n_0)}$ for a known $A_0$ defined as per the learning problem \eqref{eq:learning_problem} and $n_0$ is the noise factor sampled from $\mathcal{N}(0.5,1)$. To put it in context, the noise factor $n_0$ simply captures the notion of unreliable connectivity. In this regard, Fig.~\ref{fig:valuation_func} identifies the variability in the scaled valuation function, measured from the buyer's perspective, in terms of model precision for a scenario with two sellers having correlated information and a buyer acting as the learner. We observe the addition of a random noise factor lowers the valuation function, i.e., a negative impact of unreliable connectivity on data trading, which is quite intuitive and straightforward. However, we also see a positive contribution of obtained information on the volume of correlated data samples that in return maximizes the valuation function of the learner.\qed
		
		\item \textbf{Case II:} When $\xi_m >0,\forall m$, the learner only has access to a subset of the device's data (in the best case scenario), or just partial data (for example, the learning parameters). In this later scenario, the learner can use several distance measures, such as L2-norm, cosine similarity and so on, to figure out correlated learning parameters. 
		\item \textbf{Case III:} When $\xi_m \ge 0, \forall m$, i.e., a particular case of I and II.
	\end{itemize}
	\begin{figure*}[t!]
    \centering
\begin{tikzpicture}
  \node (p1)      at ( 0,0) [shape=circle,draw, fill=black!30] {P};
  \node (d1)    at ( -3,-1.5) [shape=circle,draw, fill=black!30,  label=below:$D_1{=\{1,2,3,4,5\}}$] {1};
  \node (d2)      at ( -1,-2.3) [shape=circle,draw, fill=black!30, ,  label=below:$D_2{=\{1,2\}}$] {2};
  
   \draw [-] (p1) to (d1);
   \draw [-] (p1) to (d2);
   
    \draw [-, thick] (1.2,1.5) to (1.2,-4.2);
  \node (p2)      at ( 5.5,0) [shape=circle,draw, fill=black!30] {P};
  \node (2d1)    at ( 3,-1) [shape=circle,draw, fill=green!50,  label=below:$D_1{=\{1,2,3,4,5\}}$] {1};
  \node (2d2)      at ( 5,-2.3) [shape=circle,draw, fill=black!30, ,  label=below:$D_2{=\{1,2\}}$] {2};
  
   \draw [-] (p2) to (2d1);
   \draw [-] (p2) to (2d2);

    \draw [dashed] (3.7,-0.1) arc [start angle=190, end angle=300, x radius=1.75cm, y radius=1cm];

  \node (label a)      at ( 5,-3.8) {$\textrm{Case I}$};  
  
  \node (p3)      at ( 10.5,0) [shape=circle,draw, fill=black!30] {P};
  \node (3d1)    at ( 8,-1) [shape=circle,draw, fill=green!50,  label=below:$D_2{=\{1,2\}}$] {2};
  \node (3d2)      at ( 10,-2.3) [shape=circle,draw, fill=black!30, ,  label=below:$D_1{=\{1,2,3,4,5\}}$] {1};
  
   \draw [-] (p3) to (3d1);
   \draw [-] (p3) to (3d2);

    \draw [dashed] (8.7,-0.1) arc [start angle=190, end angle=300, x radius=1.75cm, y radius=1cm];
    \node (label a)      at ( 10,-3.8) {$\textrm{Case II}$}; 

\end{tikzpicture}
    \caption{An illustrative experimental setup for data trading and value exchange.}
    \label{fig:experimental}
\end{figure*}
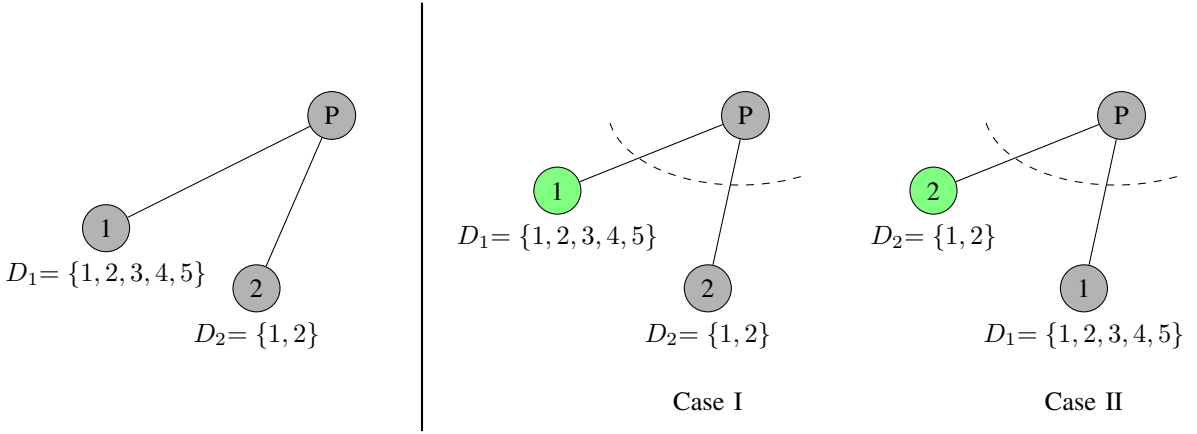
	Next, in the following, we define the relevant utility function for the learner and devices to characterize these properties in the IoT data market.
	\subsection{Utility formulations}
	The learner exploits the solution obtained from solving \eqref{eq:mse} to maximize its advantage of knowing the data types of devices for a best-suited pricing scheme. In particular, the learner aims at maximizing the following problem:
	\begin{equation}
		\sum\nolimits_{m\in\mathcal{M}} V(g(\xi_m)|D_m,a_m(t),p_m(t))-a_mp_m,
	\end{equation}
	where $V(\cdot)$ is a non-decreasing, concave valuation function evaluated at the learner knowing data types to lower down the pricing. We later discuss the details of it. On the contrary, with the given pricing, the devices intend to lower the risk of exposing their data types to the learner. An approach to address this concern while having a method that captures the concerns raised in \textbf{Q1}, \textbf{Q2} and \textbf{Q3} is the following maximization problem for each device:
	\begin{equation}
		\bigg[\sum\nolimits_{i\in\mathcal{M}\setminus m} V_i(\cdot)+a_mp_m\bigg]- \delta_m V_m(D_m,a_m(t),p_m(t)), 
		\label{eq:device_problem}
	\end{equation}
	where $V_{i\in \mathcal{M}\setminus m}(\cdot)$ is the value of added data in the data market, $\delta_m\ge0$ is the sensitivity of the data market in optimization \eqref{eq:mse} over the revealed data type for device $m$. However, solving \eqref{eq:device_problem} exerts additional communication overhead to calculate the valuation of all participating devices in the market, and therefore, is inefficient in deriving low-cost, distributed solutions to meet the research objectives. Therefore, we need to redesign the interaction scenario between devices and the learner in the data market. To this end, we develop a composite objective that stimulates the participation of the devices and brings value from the exchange of data between the devices and the learner in the data market. \rv{Our focus is to realize a trusted data market setting that brings participation of the devices with correlated data properties in a group, limiting uncontrolled competition and probable leakage of information about each other's data properties; hence, the market cannot unilaterally depress pricing.}
	
	\subsection{Multi-Agent Joint Policy (MAJP): A distributed coalition strategy}\label{subsec:majp}
	
	\rv{We start by developing the underlying distributed coalition game structure of the posed problem statement as an optimization problem. We recast the interaction between devices and the learner considering the possibility of leakage of correlation information as a multi-agent cooperative game where the payoff during coalition is allocated amongst the devices for tightening correlated information. This fundamentally means the devices with correlated information negotiate to derive a stable equilibrium solution, where both of them benefit from the data market.} Recall the data trading mechanism in Definition~\ref{def:mechanism}, this also means the coalition strategy works best for everyone's interest bringing the higher value of data, pricing, privacy and learning. Herein, we also drop the notion of time $t$ and evaluate the system for each round of interaction between the devices and the data market, which is a valid assumption to make.
	
	Let $\mathcal{M}$ denote a grand coalition and $\mathcal{S}\subseteq\mathcal{M}$ is a set of devices in coalition to protect their correlated data. In particular, a set of devices $\mathcal{S}$ agree to act as a single entity to negotiate with the value of their collective data $\mathcal{D}_\mathcal{S}$ with the platform during data trading. The value of coalition is therefore related to the pricing rate $p(\mathcal{D}_\mathcal{S})$ such that $p:P(\mathcal{M}) \times \mathcal{M} \rightarrow \mathbb{R}_{\ge 0}$, where $P(\mathcal{M})$ is the power set of $\mathcal{M}$. In what follows, we reuse the data type $\phi_m$ of device $m$, defined as a function of the importance value $\xi_m$ it allocates for the privacy of data $D_m$, as a consensus constraint on the coalition property. Consider $\mathcal{A}= a_1 \times a_2 \times \ldots \times a_m$ is the action space defining the device's joint agreement in the data trading process. Then, with reference to \eqref{eq:device_problem}, the added contribution of coalition $\mathcal{S}$ in the system can be defined by a utility function as follows.
	\begin{definition}\label{def:utility}
		For a given coalition $\mathcal{S}\subseteq\mathcal{M}$, $f_{\mathcal{S}}((\phi_j(\mathcal{S, \mathcal{A}})):\mathbb{N}\rightarrow \mathbb{R}\in[0,1]$ is a positive, concave utility function that adds return on investment for having a coalition and tightening the information about data types $\phi_j,\forall j\in \mathcal{S}$. 
	\end{definition}
	Following Definition \ref{def:utility}, we can define the coalition value instead of individual utilities as 
	\begin{align}
		v(\mathcal{S, \mathcal{A}}) &= \sum\nolimits_{j \in \mathcal{S}}a_j\Bigg[\bigg(p_j(n_j)+f_{\mathcal{S}}(\phi_j(\mathcal{S, \mathcal{A}})) \bigg) \notag\\ 
		&\quad-\bigg(p_j\cdot c(\Delta\phi_j) + c(\mathcal{S})\bigg)\Bigg],
		\label{eq:coalition_value}
	\end{align}
	where $p_j(n_j), \forall j\in\mathcal{S}$ is proportional pricing for $n_j$ data available at the devices in the coalition, $c(\Delta\phi_j)$ is the opportunity cost when leaking data type information to the learner following early trading, i.e., the learner is allowed to optimally minimize $\mathbb{E}\bigg[(\phi_m - g(\xi_m|D_m, a_m(t), p_m(t)))^2\bigg], \forall m \in \mathcal{S}$, and $c(\mathcal{S})$ is the cost of coalition defined in terms of the total power required to exchange information on correlation. More formally, we define $p_j(n_j)=p_s\bigg[\frac{n_j}{\bigcup_{j \in \mathcal{S}}n_j}\bigg]$ for a defined budget $p_s$ on the coalition $\mathcal{S}$, $c(\Delta\phi_j) = \sum\nolimits_{i \neq j}g_{ij}a_ia_j, \forall i,j \in \mathcal{S}$, where $g_{ij}$ is the normalized influence of device $i$ to $j$ due to correlation properties in the data. 
	\begin{theorem}
		For a single seller case, the optimal coalition value $v^*(\mathcal{S, \mathcal{A}})$ is proportional to the offered pricing $p^*(n)$ for trading data of samples $n$. In a multiple seller case, given a known cost of coalition $c(\mathcal{S})$, the optimal coalition value is proportional to the gain from tightening the information leakage due to data correlation and the availability of data samples itself. 
	\end{theorem}
	\begin{proof}
		The first case is simple to prove. The absence of data rivalry leads to data trading with pricing signal sufficient enough for active participation of the data seller devices, leading $c(\Delta\phi_j)$ and $c(\mathcal{S})$ to zero, i.e., $a_j=1$ and $a_i=0, \forall i\in \mathcal{S}\setminus j$. Given $c(\mathcal{S})$ and a linear pricing scheme, maximizing the coalition value corresponds to minimizing the components that capture the impact of data correlation defined as $g_{i,j}, \forall i,j \in \mathcal{S}$ between sellers pair $\{i,j\}$, i.e., $p_j\cdot c(\Delta\phi_j)$; hence, the participation of devices to add value within the coalition and maximize $f_{\mathcal{S}}((\phi(\mathcal{S, \mathcal{A}}))$ with more data samples. 
	\end{proof}
	We note that \eqref{eq:coalition_value} presents a holistic outlook to the problem that connects data value, pricing, privacy and learning in the IoT data market. In what follows, if we consider a typical learning problem \eqref{eq:learning_problem} solved by the data market via data trading, it is of particular interest to realize the data value and its impact on the learning performance for a given pricing scheme, as a usual case in the data market. That also poses a feasible approach where the platform feedback the impact of parameter dissimilarity as in the opportunity cost $c(\Delta\phi_j)$ across devices due to their data properties. This can be achieved with the following definition.
	\begin{definition}\label{def:dissimilarity}
		The devices participating in the data market exhibit parameter dissimilarity $\rho_m \ge 0, \forall m \in \mathcal{M}$ in terms of gradients on the global and local loss as $||\nabla J_m(w) - \nabla J(w)|| \le \rho_m, \forall w$.
	\end{definition}
	Then, we can derive the average data dissimilarity in the coalition $\mathcal{S}$ following the Definition \ref{def:dissimilarity} as
	\begin{equation}
		\rho_\mathcal{S} = \sum\nolimits_{j \in \mathcal{S}}a_j\Bigg[\rho_j\cdot\frac{n_j}{\bigcup_j n_j}\Bigg]. \label{eq:dissimilarity}
	\end{equation}

\begin{table}[t!]
	\begin{centering}
		\begin{tabular}{|c|c|c|}
			\hline
			\textbf{Sellers} & \textbf{Case I} & \textbf{Case II} \cr
			\hline
			$\mathcal{D}_1:\{1,2,3,4,5\}$& \{1,2,3,4,5\} & \{3,4,5\} \cr
			\hline
			$\mathcal{D}_2:\{1,2\}$& - & \{1,2\} \cr
			\hline
		\end{tabular}\\
	\end{centering}
	\caption{Illustrative example on data trading where Case I indicates seller $\{1\}$ approaching the platform first, and Case II, otherwise.}
	\label{tab:trading}
\end{table}

\begin{table}[t!]
	\begin{center}
		\begin{tabular}{|l|ll|ll|}
			\hline
			\multicolumn{1}{|c|}{}                                 & \multicolumn{2}{l|}{\textbf{Coalition}}                & 
			\multicolumn{2}{l|}{\textbf{No coalition}} \\ \cline{2-5} 
			\multicolumn{1}{|c|}{\multirow{-2}{*}{\textbf{Cases}}} & \multicolumn{1}{l|}{$c(\Delta\phi_1)$}         &   $c(\Delta\phi_2)$  &  \multicolumn{1}{l|}{$c(\Delta\phi_1)$}          & $c(\Delta\phi_2)$          \\ \hline
			Case I:                                                & \multicolumn{1}{l|}{\cellcolor{blue!25}1} & -                         & \multicolumn{1}{l|}{0}         & 1         \\ \hline
			Case II:                                               & \multicolumn{1}{l|}{-}                         & \cellcolor{blue!25}1 & \multicolumn{1}{l|}{1}         & 0         \\ \hline
		\end{tabular}
	\end{center}
	\caption{Analysis of coalition strategy on linear opportunity costs $c(\Delta\phi_j), j \in \{1,2\}$ for a unit cost per sample with $\mathcal{D}_{j|j=\{1,2,3,4,5\}}$ and $\mathcal{D}_{j|j=2}=\{1,2\}$.}
	\label{tab:opp_cost}
\end{table}
	
	\noindent\textbf{Example 2.} (Illustrative analysis of coalition strategy on opportunity costs): As an appetizer, we set the availability of two devices in the trading system. Consider two sellers with data $\{\mathcal{D}_{i|i=1,2}\}$ and the buyer setups the data market for $\{\mathcal{D}_i\cup \mathcal{D}_j \subseteq \mathcal{D} | \mathcal{D}_i \cap \mathcal{D}_j\neq \emptyset\}$ data samples, as shown in Fig.~\ref{fig:experimental}. Then, if $\mathcal{D}_1$ and $\mathcal{D}_2$, respectively, possess a subset of data samples in the market, i.e., $\mathcal{D}_1=\{1,2,3,4,5\}$ and $\mathcal{D}_2=\{1,2\}$, then for a unit monetary value on the data dissimilarity $\rho_j, \forall j \in \{1,2\}$, we have two specific cases for data trading and the involved opportunity costs. Case I: Seller $\mathcal{D}_1$ trading its data first and Case II: Seller $\mathcal{D}_2$ trading its data first. In this regard, Table \ref{tab:trading} shows the trading procedure and Table \ref{tab:opp_cost} evaluates the opportunity cost under two particular scenarios: (i) when a coalition is formed and (ii) when individual trading is performed. To elaborate, let's say Seller $\mathcal{D}_1$ considers Case I and opts out of coalition to lower its opportunity cost. Then, $\mathcal{D}_2$ can switch to the early trading scenario, i.e., Case II, to lower its own opportunity cost, consequently forcing $\mathcal{D}_1$ to form a coalition with seller $\mathcal{D}_2$. Likewise, the narratives on the stability of coalition formation under Case II. This illustrative analysis establishes the motive of devices to self-organize into coalition in the data market, which we later show is stable, to alleviate the impact of data similarity between them on the offered pricing and the aftermaths of data rivalry. \qed
	
	\begin{figure}[t!]
    \centering
    \begin{tikzpicture}
      \node (p1)      at ( 1.2,0) [shape=circle,draw, fill=black!30] {P};
      \node (d1)    at ( -2,-1.5) [shape=circle,draw, fill=black!30] {1};
      \node (d2)      at ( -1,-2.3) [shape=circle,draw, fill=black!30] {2};
      
      \node (d3) at ( 1.6,-2.2) [shape=circle,draw, fill=black!30] {3};
      \node (d4) at (3.6,-2.5) [shape=circle,draw, fill=black!30] {4};
      \node (d5) at ( 3,-3.5) [shape=circle,draw, fill=black!30] {5};
      
       \draw [-] (p1) to (d1);
       \draw [-] (p1) to (d2);
        \draw [-] (p1) to (d3);
        \draw [-] (p1) to (d4);
        \draw [-] (p1) to (d5);
    
        \draw [dashed] (0,-0.2) arc [start angle=190, end angle=300, x radius=1.75cm, y radius=1cm];
        \draw [dashed] (0,-2.8) [ rotate=-32, thick, blue] arc [start angle=0, end angle=360, x radius=1.8cm, y radius=0.8cm];
        \node (type 1)      at ( -2,-2.2) {$\Phi_1$};
        
        \draw [dashed] (4,-3.5) [ rotate=-32, thick, blue] arc [start angle=0, end angle=360, x radius=1.6cm, y radius=1.2cm];  
        \node (type 2)      at ( 2,-3.3) {$\Phi_k$};
        \node (p)      at ( 0.5,-5) [shape=circle,draw, fill=black!30, label=right:$\textrm{Platform}$] {P};

        \draw[black,very thick, dotted] (0.3,-2.5) -- (0.7,-2.6);
    
\end{tikzpicture}
    \caption{\rv{An illustrative snapshot of the distributed coalition game. Each area encircled by the dotted blue signifies a coalition group including devices as its coalition members.}}
    \label{fig:bcg}
\end{figure}
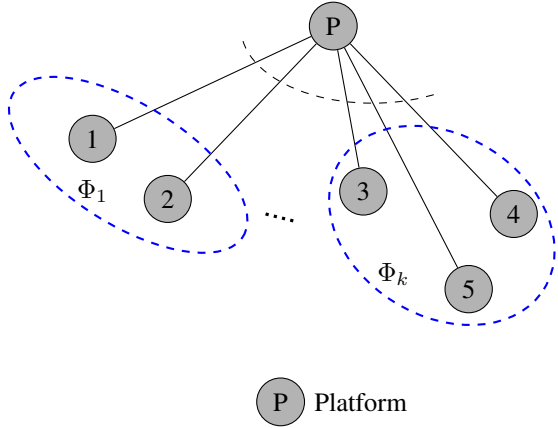
	With respect to the above analysis and its natural consequence, the data market intends to solve the following MAJP optimization problem in its general form.
	\begin{maxi!}[2]  
		{\boldsymbol{\{a_i, p_i\}_{i \in \mathcal{S}}}}                               
		{ v(\mathcal{S, \mathcal{A}})} {\label{opt:P}}{\textbf{P:}}
		\addConstraint{\sum\nolimits_{i\in \mathcal{S}}p_i = p_\mathcal{S} \label{cons1:budget}},
		\addConstraint{V_{i \in \mathcal{S}}(D_i, a_i, p_i|\phi_i) \ge 0, \forall i \in \mathcal{S}, \label{cons2:positive_value}} 	
		\addConstraint{V_{i \in \mathcal{S}}(D_i, a_i, p_i|\phi_i)\ge \notag\\
			&&&V_{i \in \mathcal{S}'}(D_i, a_i, p_i|\phi_i), \mathcal{S}' \in \mathcal{S}\setminus i,\label{cons3:valuation}} 	
		\addConstraint{a_i \in \{0,1\}, \forall i \in \mathcal{S}, \label{cons:participation}}
		\addConstraint{\rho_\mathcal{S} \le \rho, \label{cons:dissimilarlity}}
		\addConstraint{c(\mathcal{S}) \; \textrm{is bounded,} \label{cons:coalition_cost}} 
		\addConstraint{c(\Delta\phi_i)\le \phi_i^{\textrm{th}}, \forall i \in \mathcal{S}, \label{cons:opportunity_cost}} 
	\end{maxi!}
	where \eqref{cons1:budget} is the budget constraint available to be distributed amongst the members in the coalition; constraints \eqref{cons2:positive_value} and \eqref{cons3:valuation} jointly capture the positive valuation in participation; \eqref{cons:participation} defines the participation strategy of the devices, \eqref{cons:dissimilarlity} quantifies the measure of average data dissimilarity such the coalition is stable; \eqref{cons:coalition_cost} is the accepted tolerance on the cost of the coalition, and \eqref{cons:opportunity_cost} is the bound on individual opportunity cost of the members in the coalition. We notice, the optimization problem \textbf{P} is hard to solve and mostly intractable due to (i) binary constraints, (ii) the stability of the mechanism due to the coupling in data types and valuation for the unknown heterogeneity in data distribution, for a large number of devices, and (iii) private cost information. To address the technical challenges defined for solving \textbf{P}, in the following, we recast the market design so as to offer a novel cooperation protocol that mitigates the tension between data sellers and the data market using the pricing signal for the exchange of value and data. In particular, we characterize a subset of the coalition formed accordingly to the data types and then derive an association and pricing scheme based on the properties of devices associated with the particular coalition. Note that the obtained solution is sub-optimal but a low-complexity alternative to address the aforementioned challenges. 
	
	\section{Market Design: MAJP as a Coalition Game}\label{sec: market design}
	In the best interest of the devices (``agents"), we observe, as illustrated in the experimental (subsection~\ref{subsec:majp}), devices having \textit{similar data types form coalition} to maximize their utility and reach stability in the data market offering transferable utility (TU). Hence, we formalize the elegant framework of distributed coalition games, the \emph{Hedonic game} \cite{bogomolnaia2002stability, apt2009generic} to solve the problem of distributed coalition with the objective of maximizing \textbf{P}. In fact, it is intuitive that the devices have individual preferences to form coalition groups with similar data types, which is a common concept for coalition-based games \cite{apt2009generic}. This captures two necessary conditions to design the Hedonic game: (i) the payoff of the devices is defined only based on the other members in the coalition, and (ii) the coalition structure is the direct consequence of the preference profiles of each device in the coalition.
	\begin{Remark}
		The interaction between a single learner and a finite set of devices (sellers) in coalition upon the pricing signal and the value of data exchange together form a coalition game to protect their correlated data. Formally, the game is characterized to capture the coalition strategies of devices and the update in pricing signals of the learner towards maximization of the coalition  during participation in the data market. 
	\end{Remark}
	
	\begin{definition}\label{def:coalition}
		A coalition partition is defined as the set $\Pi = \{\mathcal{S}_1, \mathcal{S}_2,\ldots, \mathcal{S}_K\}$ dividing the total set of devices $\mathcal{M}$ in the system such that $\mathcal{S}_k\subseteq\mathcal{M}$ and $\cup_{k=1}^K\mathcal{S}_k=\mathcal{M}$, where $\mathcal{S}_k$ are the coalitions sets based on device type $k\in\{1,2,\ldots,K\}$. Then, we have the following preference definition for the device participating in the coalition \cite{bogomolnaia2002stability}. 
	\end{definition}
	
	\begin{algorithm}[t!]
		\caption{\strut MAJP Solution with Coalition Formation}
		\label{alg:game}
		\begin{algorithmic}[1]
			\STATE{\textbf{Initialization:} Partition $\Pi_{\textrm{initial}}$ with devices in set $\mathcal{M}$ having a total of $D$ data samples at $t\in\mathcal{T}$.}
			\STATE{\textbf{Output:} Stable partition $\Pi_{\textrm{final}}$, participation vector $\textbf{a} = \{a_1, a_2,\ldots,a_m\},\forall m \in \mathcal{M}$, pricing signal $\textbf{p}=\{p_{\mathcal{S}_1},p_{\mathcal{S}_1},\ldots,p_{\mathcal{S}_k}\}$, $\mathcal{S}_k \in \Pi_{\textrm{final}}$}.
			\STATE{\textbf{Phase I}: \\Private discovery of device types;\\
				Execute Algorithm~\ref{alg:mpc};}
			\STATE{\textbf{Phase II}:\\ 
				Distributed coalition formation;}\\
			\REPEAT
			\FORALL{device $m\in\Pi_{\textrm{initial}}$}
			\STATE Randomly select two coalitions;
			\STATE Evaluate the preference function for the given coalition with \eqref{eq:pref_func} and preference profile (Definition 4);
			\STATE Invoke switch operations between coalition groups, comparing possible payoffs;
			\STATE Add device to the observed coalition sets in $\Pi_{\textrm{initial}}$;
			\STATE Repeat evaluation of preferences on different coalition groups (line 6) until no further switch operations exist;
			\ENDFOR
			\UNTIL{$\Pi_{\textrm{final}}$ is reached;}
			\STATE{\textbf{Phase III}: Computation of coalition value using \eqref{eq:coalition_value};}
			\STATE{Evaluate the final pricing signal $\textbf{p}$;}
		\end{algorithmic}
		\label{Algorithm}
	\end{algorithm}
	\begin{definition}\label{def:preference}
		The preference profile of any device $m$ is defined by the relation or an order $\succeq_m$ that is a complete, reflexive, and transitive binary relation over the set $\{\mathcal{S}_k\subseteq\mathcal{M}: m \in \mathcal{S}_k\}$.
	\end{definition}
	
	Following Definition~\ref{def:preference}, we have for any pair of coalition sets $\mathcal{S}_1$ and $\mathcal{S}_2$, $\mathcal{S}_1 \succeq_m \mathcal{S}_2$ means device $m \in \mathcal{M}$ prefers coalition $\mathcal{S}_1$ (or least equally prefers both), than $\mathcal{S}_2$. 
	
	Then, formally, with the given set of devices $\mathcal{M}$ as players and their preference profiles $\succeq_m, \forall m \in \mathcal{M}$, a Hedonic coalition game can be defined as follows.
	\begin{definition}
		A Hedonic game is defined by the pair of set of players (i.e., the devices) and their preference profiles $(\mathcal{M}, \succ)$.
	\end{definition}
	Once the coalition $\mathcal{S}_k$, for illustration, as shown in Fig.~\ref{fig:bcg}, with TU \eqref{eq:coalition_value} is agreed upon by the devices, the coalition utility can be divided amongst the devices as the payments, quantified in the form of contract. In our formulation, we define the value of coalition as the coupling between the obtained overall revenue due to participation in the data market and the consequence of limited information leakage due to correlation properties amongst the device's data. Particularly, the payment under \textit{contracts} for device $m \in \mathcal{M}$ is defined as $p_m(\mathcal{S}_k)=p_{\mathcal{S}_k}\bigg[\frac{n_m}{\bigcup_{m \in \mathcal{S}_k}n_m}\bigg]$ for a obtained revenue $p_{\mathcal{S}_k}$ due to coalition, where $\sum_{\mathcal{S}_k \in \Pi} p_{\mathcal{S}_k} = p_\mathcal{S}$ on the coalition $\mathcal{S}_k$, using the earlier defined concepts.
	
	Next, following Definition~\ref{def:preference}, we evaluate the preference profile of the devices as follows. Let's define $\mathcal{S}_{\Pi}(m)$ is the coalition set where the device $m$ should belong following its type, i.e., $m \in \mathcal{S}_k$ such that $\mathcal{S}_{\Pi}(m) = \mathcal{S}_k \in \Pi $. Then, as explained in Section~\ref{sec:preliminaries}, the preference of devices is defined as $\mathcal{S}_1 \succeq_m \mathcal{S}_2 \Leftrightarrow u_m(\mathcal{S}_1) \succeq_m u_m(\mathcal{S}_2), \forall m \in \mathcal{M}$, where $u_m: 2^{\mathcal{K}}\rightarrow \mathbb{R}$ is the preference function of any device $m$ such that
	\begin{equation}
		u_m(\mathcal{S}) =
		\begin{cases}
			p_m(\mathcal{S}), \; \; \text{if $m \in \mathcal{S}_{\Pi}(m)$},\\
			0, \; \; \text{otherwise}.
		\end{cases}
		\label{eq:pref_func}
	\end{equation}
	
	In doing so, the devices verify the conditions of bound on their individual opportunity costs and the measure of parameter dissimilarity, as defined in Definition \ref{def:dissimilarity}.
	Then, it is quite straightforward to have the preference function as \eqref{eq:pref_func}. The devices in coalition get benefit from the TU obtained following the contracts mechanism, as discussed before, where two specific economic properties (or conditions) are satisfied: (i) \textit{Individual Rationality (IR)}, a condition that captures the motive behind devices undergoing distributed coalition with a positive return on investment, (ii) \textit{Incentive Compatibility (IC)}, a condition that ensures devices get to maximize their utilities, as in the form of obtained payments, if they act as per their type.
	
	\begin{figure}[t!]
		\centering
		\includegraphics[width=\linewidth]{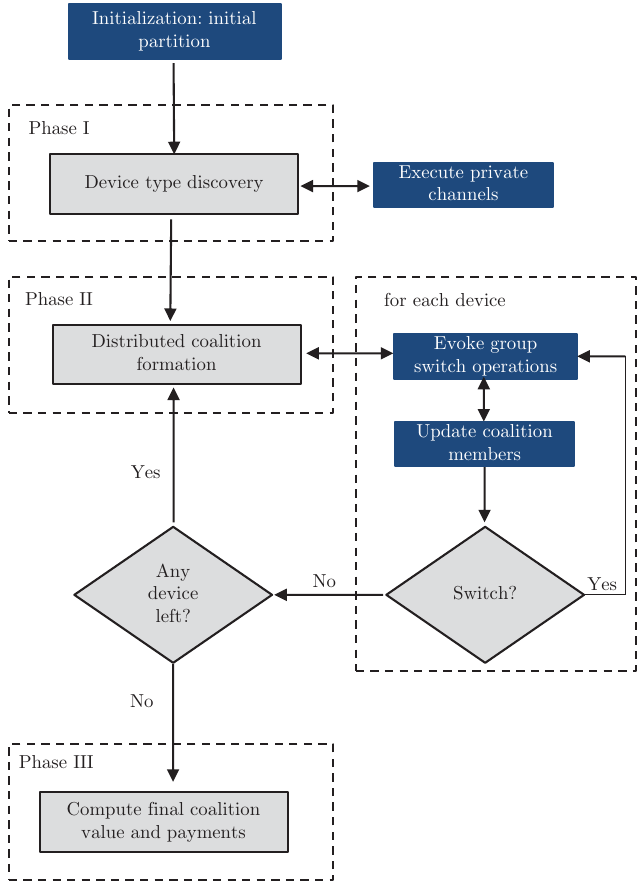}
		\caption{Execution flowchart of MAJP solution.}
		\label{fig:flowchart}
	\end{figure}
	In Algorithm~\ref{alg:game}, we develop a solution approach to the optimization problem \textbf{P}, where the objective is to derive coalition partitions following device types that uniquely maximize the overall coalition value. In particular, we adopt the modified merge-and-split algorithm \cite{apt2006stable, apt2009generic, saad2009coalitional}, which works in an iterative manner on two coalitions at a time, and design an intuitive solution that reaches stable partitions. To that end, the following remark is useful.
	\begin{Remark}
		While satisfying IR and IC constraints, we show the developed MAJP framework offering the distributed coalition solution with payment under contracts to realize an IoT data market for model training yields the following desirable properties.
	\end{Remark}
	
	\begin{enumerate}[label=(\roman*)]
		\item Budget balance -- Following Definition~\ref{def:preference}, for any device $m \in \mathcal{S}_k$, the sharing of total per group budget allocation $p_\mathcal{S}$ between the coalition members is proportional to the value of their participation such that $\sum \nolimits_{\mathcal{S}_k \in \Pi}p_{\mathcal{S}_k} = p_\mathcal{S}$ is satisfied.
		\item Linearity (within the coalition) -- The linearity property, basically implying the revenue allocation for the exchange of data $\mathcal{D}_1$ and $\mathcal{D}_2$ from devices $\{1\}$ and $\{2\}$, respectively, is the same as any one device trading $\mathcal{D}_1 \cup \mathcal{D}_2$, comes as a direct consequence of proportional payoff within the group. 
		\item Truthfulness -- The mechanism is truthful such that there exists a non-negative payoff proportional to the data contribution, as defined in the coalition \emph{contracts}, only if the device form coalition as per its data type and the linearity. It readily follows the definition of the preference profile of any device in \eqref{eq:pref_func} and the designed corresponding individual utility function, which penalizes untruthful reporting.    
		\item Symmetry -- The output of the mechanism is invariant to any permutation of participation during coalition formation. The proof follows for a finite set of devices, given a mechanism satisfying properties (ii) and (iii). 
	\end{enumerate}
	
	(We rely on the above proof sketch and omit further details on analytical proof to support Remark 2, as it is straightforward with the given explanation.)
	
	Algorithm~\ref{alg:game} operates in three phases, as shown in Fig.~\ref{fig:flowchart}; \textbf{Phase I:} device type discovery, \textbf{Phase II:} distributed coalition formation, and \textbf{Phase III:} computation of final coalition value and payments. In Phase I, the devices exploit a secured private channel to compare their type with available proxy device type sets. To allow its successful execution, we assume devices can ``ping" each other through a broadcast network or use beacons during the discovery phase to evaluate individual devices types, similar to \cite{clarke2010private}. In this regard, the cost of the coalition can be explicitly defined as units per round of communication. In Phase II, we employ the execution of split-and-merge algorithm \cite{apt2009generic}, where the devices are allowed to make a switch between two coalition groups at a time, following their preference profile on the device types and utility functions, that maximizes their payoff. Following Definition \ref{def:dissimilarity}, the bound on individual opportunity costs, and the amount of information leakage, the devices perform coalition switch operations as per their preference profile. This iterative procedure eventually leads to a stable coalition structure, as proven in \cite{apt2009generic,saad2009coalitional}. Finally, in Phase III, once the stable partition $\Pi_{\textrm{final}}$ is reached, the coalition value is calculated further to define the final pricing signal for the data trading.
	
	\begin{algorithm}[t!]
		\caption{\strut Private Discovery of Device Types}
		\label{alg:mpc}
		\begin{algorithmic}[1]
			\STATE{Begin with initial partition $\Pi_{\textrm{initial}}$ of $K$ with proxy data type sets in $\mathcal{K}$}
			\STATE{Permute devices within $\mathcal{S}_k, k\in \{1,2,\ldots, K\}$ over available private channel to evaluate device data type.}
			\STATE{Return evaluated device types to individual devices $m\in\mathcal{M}$.}
		\end{algorithmic}
	\end{algorithm}
	
	\subsection{Optimality and convergence analysis}
	The solution derived following Algorithm~\ref{alg:game} is low-complexity, sub-optimal in solving \textbf{P}. Therein, we first analyze the convergence of the algorithm to a stable coalition. For that, we state the following theorem based on the concept of weak Pareto optimal \cite{jorswieck2011stable} leveraged from the stability in the matching problem.
	\begin{theorem}(Informal) \label{theo:convergence}
		Algorithm~\ref{alg:game} converges to a local maximal value following the definition of weak Pareto optimality and the existence of a stable coalition.
	\end{theorem}
	\input{figinfoleak}
	\begin{proof}
		We note that the objective $v(\mathcal{S, \mathcal{A}})$ captures the overall utility with the obtained coalition. Then, at each iteration, say $t$, indicated for the execution of lines 6 -- 12, we have an outcome as the transient partition set $\Pi_{(t)}$ that provides associations of devices to a particular coalition. Then, each $t$-operation gives a guarantee in the relation 
		\begin{equation}
			v(\mathcal{S}_{\Pi_{(t)}}, \mathcal{A}^t) \ge v(\mathcal{S}_{\Pi_{(t-1)}}, \mathcal{A}^{(t-1)})
		\end{equation}
		following the accept/reject operation of participation variable $a_i, \forall i \in \mathcal{S}$ for the individual utility maximization. We observe the maximization objective of \textbf{P} is captured using a non-decreasing function under the participation decision as accept/reject operation of Algorithm~\ref{alg:game}. Then, the MAJP algorithm converges to a local maximal point of the problem \textbf{P}.
	\end{proof}
	Next, we have a quick walk-through example to reflect the operation of Algorithm~\ref{alg:game} on optimality and convergence following Theorem~\ref{theo:convergence}. First, define a coalition game $(v, \mathcal{M})$, where $\mathcal{M}\in \{1,2,3,4\}$ following a known optimal device to device's data type partition mapping rule $\mathcal{A}^*$ for a subset of partition $\mathcal{S}_k, k \in \{1,2,\ldots, K\}$ as
	\[ \mathcal{A}^* :=
	\bordermatrix{ & \mathcal{S}_1 & \mathcal{S}_2 & \mathcal{S}_3 \cr
		\{1\} & 1 & 0 & 0 \cr
		\{2\} & 0 & 1 & 0 \cr
		\{3\} & 1 & 0 & 0 \cr
		\{4\} & 0 & 0 & 0 }, \qquad
	\]
	Following which, let us assume having an oracle defining corresponding $v$ as
	\begin{equation}
		v(\mathcal{S}, \mathcal{A}^*) :=
		\begin{cases}
			1, \; \; \text{if $\mathcal{S} = \{1,2\}$},\\
			2, \; \; \text{if $\mathcal{S} = \{1,3\}$},\\
			0, \; \; \text{otherwise}.
		\end{cases}
		\label{eq:pref_func}
	\end{equation}
	Then, we begin with an initial partition as $\{1\},\{2\},\{3\},\{4\}$. Following Algorithm~\ref{alg:game}, we have a switch operation between coalition groups, resulting $\{\{1, 2\}, \{3\},\{4\}\}$ or $\{\{1,3\}, \{2\},\{4\}\}$. This leads to obtaining an overall utility value of $1$ and $2$, respectively, with a unique coalition group offering a higher utility and no further switch operation. Furthermore, we observe departure from the related coalition group only lowers the individual payoff for the participating device within the group - otherwise, deviation from the local optimal in each execution of the merge and split rule. 
	
	The problem of such a class, however, is mathematically intractable to derive an analytical optimality gap \cite{bogomolnaia2002stability, shehory1997multi, saad2009coalitional}. Hence, following the above analysis, we observe the sequence of $\{ v(\mathcal{S}_{\Pi_{(t-1)}}, \mathcal{A}^{t-1}), v(\mathcal{S}_{\Pi_{(t)}}, \mathcal{A}^{(t)}), v(\mathcal{S}_{\Pi_{(t+1)}}, \mathcal{A}^{(t+1)}), \ldots, \}$ has a ratio bound $\alpha = \frac{v(\mathcal{S}, \mathcal{A})}{v(\mathcal{S}^*, \mathcal{A}^*)}$ that grows logarithmically with the size of the coalitions, i.e.,
	\begin{equation}
		\alpha \le \sum\nolimits_{i=1}^{\max(|\mathcal{S}|)}\frac{1}{i},
	\end{equation}
	where the bound is derived accordingly to the concept of set cover theory, similar to \cite{shehory1996formation, shehory1997multi}, for the value at optimal coalition configuration $v(\mathcal{S}^*, \mathcal{A}^*)$ in a multi-agent coordination setup through coalition formation. To that end, in the following, the complexity analysis of the proposed algorithm is discussed.
	\subsection{Complexity analysis}  
	The complexity analysis of the proposed approach is done following the ``propose and swap"  method of the stable matching algorithm with externalities \cite{roth1992two}. Particularly, the impact of preference profiles with the number of devices and their interactions during the distributed coalition formation with the merge-and-split algorithm is an instance of the propose and swap method, where devices opt to join the coalition based on their type so as to maximize formulated individual utility-leading to stability. We begin with two randomly selected coalitions; hence, $\binom{M}{2}$ defines the number of possible switches between coalition groups and $M \times \Pi_{\textrm{initial}}$ number of possible options to split and merge as a member of coalition group. Following which, this leads to a sub-linear complexity $\mathcal{O}( M\Pi_{\textrm{final}}\log(M \times\Pi_{\textrm{initial}}))$ with the number of devices and coalitions formed, similar to the analysis done in \cite{roth1992two}. In the following, we provide simulation results to evaluate the performance analysis of Algorithm \ref{alg:game}. 
	\section{Numerical Results}\label{sec:simulations}
	\rv{In this section, we evaluate, compare, and validate the performance of the proposed market model with intuitive baselines. To begin with, first, we conduct statistical analysis to measure the information leakage due to data correlation in its simplified version. For this, we consider a few seller devices available in the market generating explanatory data samples for trading, with the quality of data defined for the model's performance at the learner, as in \cite{zhan2020learning}. Second, we measure the impact of information leakage on data valuation and, with experimental, show the impact of data rivalry on value depression. \rrv{Finally, we conduct numerical evaluations of the proposed solution using a simple Python \cite{van2014python} simulator.} Particularly, we follow a linear pricing scheme to compare the performance of our proposed approach in the IoT data market.}	

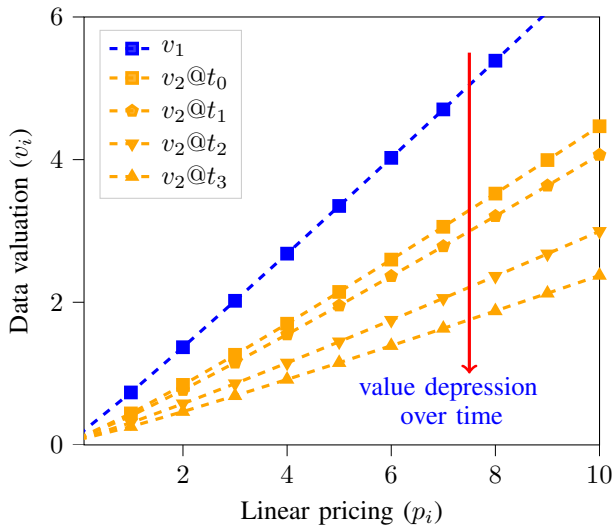
\begin{figure}
    \centering
    \begin{tikzpicture}

    \definecolor{darkgray176}{RGB}{176,176,176}
    \definecolor{lightgray204}{RGB}{204,204,204}
    \definecolor{orange}{RGB}{255,165,0}
    \definecolor{steelblue31119180}{RGB}{31,119,180}
    
    \begin{axis}[
    legend cell align={left},
    legend style={
      fill opacity=0.8,
      draw opacity=1,
      text opacity=1,
      at={(0.03,0.97)},
      anchor=north west,
      draw=lightgray204
    },
    tick align=outside,
    tick pos=left,
    x grid style={darkgray176},
    xlabel={Linear pricing ($p_i$)},
    xmin=0.1, xmax=10,
    xtick style={color=black},
    y grid style={darkgray176},
    ylabel={Data valuation ($v_i$)},
    ymin=0, ymax=6,
    ytick style={color=black}
    ]
    \addplot [very thick, blue, dashed, mark=square*, mark size=2, mark options={solid}]
    table {%
    0 0.142931959013161
    1 0.73545004143538
    2 1.36963765942085
    3 2.02049727390927
    4 2.68172368287484
    5 3.35039932353803
    6 4.02485039079282
    7 4.70399619523011
    8 5.38708400749361
    9 6.07356092743565
    10 6.76300452279839
    11 7.45508207943971
    12 8.14952512736287
    13 8.84611271902665
    14 9.54466001068255
    15 10.2450102063558
    16 10.9470287161596
    17 11.6505988200322
    18 12.3556183832809
    19 13.061997324597
    };
    \addlegendentry{$v_1$}
    \addplot [very thick, orange, dashed, mark=square*, mark size=2, mark options={solid}]
    table {%
    0 0.0931817450301187
    1 0.441534799419804
    2 0.841565250982612
    3 1.2633935957762
    4 1.69904036738198
    5 2.14473763361736
    6 2.59829606922352
    7 3.05828897016893
    8 3.52371532963953
    9 3.99383546157928
    10 4.4680813454031
    11 4.94600362528921
    12 5.42723830235922
    13 5.91148477013806
    14 6.39849075482073
    15 6.88804165182233
    16 7.37995276795111
    17 7.87406354554263
    18 8.37023317554626
    19 8.86833720704793
    };
    \addlegendentry{$v_2@t_0$}
    \addplot [very thick, orange, dashed, mark=pentagon*, mark size=2, mark options={solid}]
    table {%
    0 0.0892561318455625
    1 0.405940726745276
    2 0.769604773620556
    3 1.15308508706928
    4 1.54912760671089
    5 1.95430693965215
    6 2.3666327902032
    7 2.78480815469903
    8 3.20792302694502
    9 3.63530496507207
    10 4.06643758673009
    11 4.50091238662655
    12 4.9383984566902
    13 5.37862251830733
    14 5.82135523165521
    15 6.26640150165666
    16 6.7135934254101
    17 7.16278504140239
    18 7.61384834140569
    19 8.06667018822539
    };
    \addlegendentry{$v_2@$$t_1$}
    \addplot [very thick, orange, dashed, mark=triangle*, mark size=2, mark options={solid,rotate=180}]
    table {%
    0 0.0787878300200792
    1 0.311023199613202
    2 0.577710167321741
    3 0.858929063850803
    4 1.14936024492132
    5 1.44649175574491
    6 1.74886404614901
    7 2.05552598011262
    8 2.36581021975968
    9 2.67922364105285
    10 2.99538756360206
    11 3.31400241685947
    12 3.63482553490614
    13 3.95765651342537
    14 4.28232716988049
    15 4.60869443454822
    16 4.93663517863407
    17 5.26604236369509
    18 5.59682211703084
    19 5.92889147136529
    };
    \addlegendentry{$v_2@$$t_2$}
    \addplot [very thick, orange, dashed, mark=triangle*, mark size=2, mark options={solid}]
    table {%
    0 0.0727272342263783
    1 0.256070947063054
    2 0.466613289990848
    3 0.688628208303265
    4 0.917915982832619
    5 1.15249349137756
    6 1.39120845748606
    7 1.63330998429944
    8 1.87827122612607
    9 2.12570287451541
    10 2.37530597126479
    11 2.62684401331011
    12 2.88012542229432
    13 3.13499198428319
    14 3.39131092358986
    15 3.6489692904328
    16 3.907869877869
    17 4.16792818186454
    18 4.42907009239277
    19 4.6912301089726
    };
    \addlegendentry{$v_2@$$t_3$}
    \addplot [very thick, steelblue31119180, forget plot]
    table {%
    7.2 1
    };
    \draw (axis cs:5.2,0.7) node[
      anchor=base west,
      text=blue,
      rotate=0.0
    ]{value depression};
    \draw (axis cs:6,0.3) node[
      anchor=base west,
      text=blue,
      rotate=0.0
    ]{over time};
    \draw[->,draw=red, very thick] (axis cs:7.5,5.5) -- (axis cs:7.5,1);
    \end{axis}
    
    \end{tikzpicture}

    \caption{\rv{Evaluation of value depression with the offered linear pricing for $i \in \{1, 2\}$}.}
    \label{fig:datavaluation}
\end{figure}
	\subsection{Analysis on data correlation}
	In this subsection, we conduct the statistical analysis on information leakage for a three seller $\mathcal{M} = \{1, 2, 3\}$ and one buyer scenario in the developed IoT data market. For simplicity, we consider each device's data provides an equal marginal contribution to the learner, and the data samples are revealed in a sequence. As we talked about the FL setting, the contribution of data samples is, in fact, reflected in the gradient (or parametric) response provided by the individual sellers. Based on this analysis, we show the impact of data correlation on information leakage. Note that the absence of this simplification won't alter the result, as the framework well-captures individual contribution of data samples in improving the learner's model following \eqref{eq:type}. We associate each device with a random variable and generate a synthetic dataset using Gaussian process priors; the objective is limited to finding the joint probability distribution, and the data type is quantified following a uniform $k$-level quantization on the model performance, i.e., in the order of contribution in the model improvement, as defined in \eqref{eq:type}. The definition of average data similarity in \eqref{eq:dissimilarity} is used to perform such quantization equivalently, satisfying the constraint in \textbf{P}. For instance, using a three sellers setup represented with random variables (RVs) $X$, $Y$, and $Z$, respectively, we model $Z = 0.5X+Y$, with X$\sim \mathcal{N}(\mu_X, \sigma_X)$ and Y$\sim \mathcal{N}(\mu_Y, \sigma_Y)$, where the pairs $(\mu_X, \sigma_X)$ and $(\mu_Y, \sigma_Y)$ defines mean and standard deviation of the corresponding RV. Then, as shown in Fig.~\ref{fig:infoleak}, the buyer learns data correlation $r_{i,j}, \forall i,j \in \mathcal{M}$ and reaches convergence while obtaining device data type information with a precision of $5 \times 10^{-4}$. In practice, the sellers can use a batch of data samples; hence, the time step required is much less than the approach where each sample is revealed. With this, we next evaluate the impact of information leakage on data valuation.  
	\subsection{Value depression with information leakage}
	\rv{We use the valuation function following linear pricing models defined as $V = 1/(1+g)\bigg[\sum\nolimits_{i=1}^{M}\bigg(\frac{D_i}{\sum\nolimits_j D_j}\bigg)^{1-bp}\bigg]^{1/b}$ for each device, with information leakage factor $0\le g \le 1$ and a positive weight factor $b>0$ capturing the characteristics of valuation function. For the IoT data market setup, we use the statistical model defined in Section~\ref{sec:simulations}-A. We set $b=0.1$ and the range of offered pricing in $[1,10]$ monetary units. We then first evaluate the influence of the correlated data of seller $\{2\}$ on the data valuation of seller $\{1\}$ without the execution of the proposed coalition solution approach. Fig.~\ref{fig:datavaluation} shows the data valuation is proportional to the number of data samples, which is intuitive, and also to the offered pricing signal; however, it drops significantly as the learner identifies data properties between the sellers. Interestingly, as shown in Fig.~\ref{fig:datavaluation}, the price depression is prominent for the seller $\{2\}$, when seller $\{1\}$ is given the competitive advantage of arrival in the market. This is obvious given the characterization of data properties between $\{2\}$ and $\{1\}$ in Fig.~\ref{fig:infoleak} in quantifying their data type. In principle, the learner is devaluating the data for the seller $\{2\}$ due to information on its statistical properties by the traded data of the seller$\{1\}$. Next, we evaluate individual utilities of the devices following our proposed MAJP solution approach and provide further analysis of the developed coalition strategy. With the MAJP solution approach, the devices with correlated data samples undergo coalition to interact with the learner while limiting information leakage, avoiding data rivalry in participation, and consequently, data value depression. }
	
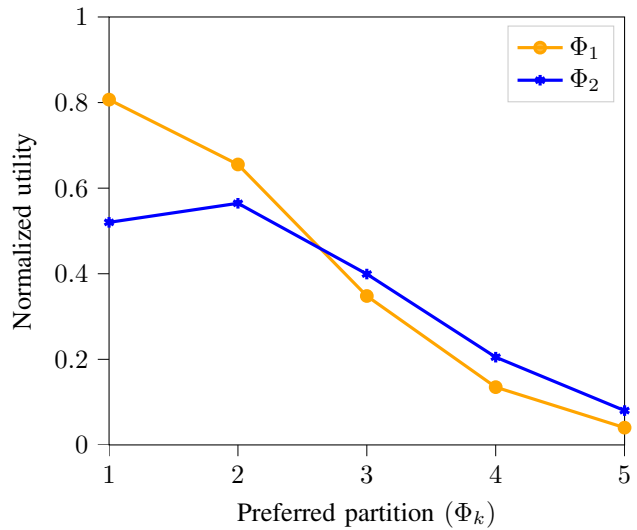
\begin{figure}
    \centering
    \begin{tikzpicture}

\definecolor{darkgray176}{RGB}{176,176,176}
\definecolor{lightgray204}{RGB}{204,204,204}
\definecolor{orange}{RGB}{255,165,0}

\begin{axis}[
legend cell align={left},
legend style={fill opacity=0.8, draw opacity=1, text opacity=1, draw=lightgray204},
tick align=outside,
tick pos=left,
x grid style={darkgray176},
xlabel={Preferred partition\(\displaystyle \ (\Phi_k)\)},
xmin=1, xmax=5,
xtick style={color=black},
y grid style={darkgray176},
ylabel={Normalized utility},
ymin=0, ymax=1,
ytick style={color=black}
]
\addplot [very thick, orange, mark=*, mark size=2, mark options={solid}]
table {%
0 0.489323590486531
1 0.806868212431536
2 0.655351198778265
3 0.348015402318706
4 0.135002623993252
5 0.0403478060332176
6 0.00947959769992077
7 0.00172277863784651
8 0.000217429349861534
9 7.49581760932939e-06
10 -5.49490248344926e-06
11 -2.08150752414138e-06
12 -4.99326155173858e-07
13 -9.64494986432398e-08
14 -1.61004237288863e-08
15 -2.39954066690391e-09
16 -3.25295028337724e-10
17 -4.06067564158626e-11
18 -4.70863641853561e-12
19 -5.105805895941e-13
20 -5.20487034369133e-14
};
\addlegendentry{$\Phi_1$}
\addplot [very thick, blue, mark=asterisk, mark size=2, mark options={solid}]
table {%
0 0.235387976677009
1 0.520087224456387
2 0.564630634025314
3 0.399293531233454
4 0.2050384588203
5 0.0802592303098457
6 0.0241663463373089
7 0.00531833811149358
8 0.000626824451680268
9 -0.000113425487911447
10 -0.000100242219766853
11 -3.91200153536011e-05
12 -1.15536891735244e-05
13 -2.87641829776577e-06
14 -6.30530945166194e-07
15 -1.24569294288559e-07
16 -2.25078277681385e-08
17 -3.75760342701318e-09
18 -5.8406002531129e-10
19 -8.50303533894261e-11
20 -1.16513379310524e-11
};
\addlegendentry{$\Phi_2$}
\end{axis}

\end{tikzpicture}
    \caption{Impact of coalition strategy as per individual data type on normalized utility of devices.}
    \label{fig:individual_utility}
\end{figure}
	\subsection{Analysis of coalition strategy}
	\rv{We begin with the evaluation and analysis of the coalition strategy developed through the MAJP solution. Fig.~\ref{fig:individual_utility} shows the impact of switch operation between coalition groups on the normalized utility for each device and the identified data types. As discussed before, while the data similarity constraints and the cost of the coalition are satisfied, any deviation of the sellers to the group different from their true data type consequently lowers its utility. This hints the seller will opt to join an appropriate coalition and undergo data trading in a group so as to maximize their utilities. Such strategies of distributed devices lead to a stable solution, consequently, as argued in Theorem 2 and validated in Fig.~\ref{fig:individual_utility}. Another observation we have from the figure is as the data are not perfectly correlated, the impact of information leakage won't penalize the normalized utility value to zero but only lowers it. Furthermore, interestingly, we have an intuitive result in Fig.~\ref{fig:individual_utility} -- i.e., joining a nearby coalition group is more beneficial for the sellers as the impact of information leakage is higher otherwise. This goes along in line with our prior analysis of the system model.} In Fig.~\ref{fig:complexity}, we validate the sub-linear complexity of the MAJP solution. For this, we set $\Pi\in\{1,2,3,4,5\}$ and include the number of devices per group as 1, 10, 20, 30, and 40, respectively. We shuffle the devices and their data type and compare the execution of the MAJP solution with the Optimal \cite{papadimitriou1998combinatorial}. The combinatorial nature of the coalition formation with the increased number of partitions and the number of associated devices results in the Optimal solution being computationally expensive as compared to the proposed MAJP solution. 
	

\begin{figure}
\centering
\begin{tikzpicture}

\definecolor{darkgray176}{RGB}{176,176,176}
\definecolor{lightgray204}{RGB}{204,204,204}
\definecolor{orange}{RGB}{255,165,0}

\begin{axis}[
legend cell align={left},
legend style={fill opacity=0.8, draw opacity=1, text opacity=1, draw=lightgray204},
tick align=outside,
tick pos=left,
x grid style={darkgray176},
xlabel={Number of devices (\(\displaystyle M\))},
xmin=-0.2, xmax=4.2,
xtick style={color=black},
xtick={1,2,3,4},
xticklabels={10,20,30,40},
y grid style={darkgray176},
ylabel={Execution time (ms)},
ymin=0, ymax=600,
ytick style={color=black}
]
\addplot [very thick, orange, mark=asterisk, mark size=2, mark options={solid}]
table {%
0 26.1512051684126
1 37.8442314920658
2 104.10515772825
3 191.222296543881
4 291.338011818933
};
\addlegendentry{Optimal}
\addplot [very thick, blue, mark=*, mark size=2, mark options={solid}]
table {%
0 11.9879228546131
1 20.410515772825
2 29.5135348519869
3 39.1222296543881
4 49.1338011818933
};
\addlegendentry{MAJP}
\end{axis}
\end{tikzpicture}
\caption{Computational complexity in terms of the execution time (in $ms$).}
\label{fig:complexity}
\end{figure}
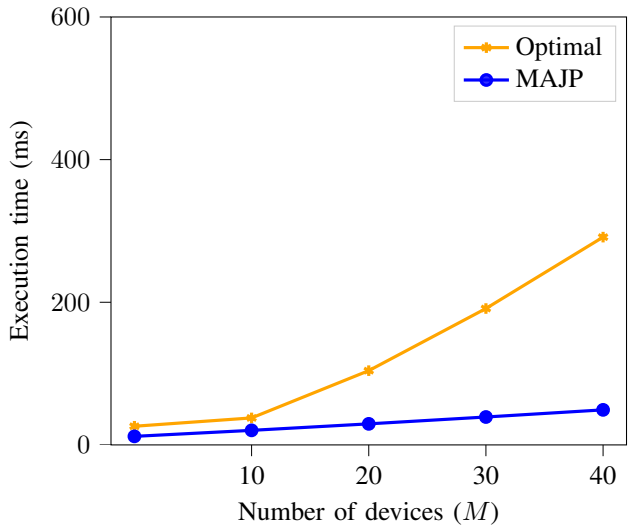
	Next, we consider the following intuitive baseline along to show the gain of adopting a coalition strategy than individual interaction with the learner under a scenario with information leakage.
	\begin{itemize}
		\item Non-cooperative: The learner exploits data properties between the sellers and imposes price depression. 
		\item MAJP solution (Cooperative): The pricing allocation follows the proposed solution approach in Algorithm~\ref{Algorithm}. 
	\end{itemize}
	For this evaluation, we reuse the linear pricing scheme with a log-concave utility on the coalition strategy adopted by devices with a similar data type to lower information leakage, as illustrated in Fig.~\ref{fig:infoleak}. For simplicity, we set $\Pi\in\{1,2\}$ and include the number of devices per group as 10, 20, 30, and 40. The results are then obtained following Monte-Carlo simulations to check and validate the consistency of the obtained results. In Fig.~\ref{fig:comparision}, we observe the proposed MAJP solution provides a gain of up to $32.72\%$ while imposing collaborative interaction between the devices with a similar data type. Interestingly, we also observe an almost flat payoff when devices opt for a non-cooperative strategy. This is reasonable given the value of information leakage with a fixed similarity in the number of data samples across devices. In this manner, the sellers benefit from the coalition to tackle price depression and information leakage to setup a trusted IoT data market. 
\begin{figure}
    \centering
    \begin{tikzpicture}

\definecolor{darkgray176}{RGB}{176,176,176}
\definecolor{lightgray204}{RGB}{204,204,204}
\definecolor{orange}{RGB}{255,165,0}

\begin{axis}[
legend cell align={left},
legend style={
  fill opacity=0.8,
  draw opacity=1,
  text opacity=1,
  at={(0.03,0.97)},
  anchor=north west,
  draw=lightgray204
},
tick align=outside,
tick pos=left,
x grid style={darkgray176},
xlabel={Number of devices (\(\displaystyle M\))},
xmin=-0.15, xmax=3.15,
xtick style={color=black},
xtick={0,1,2,3},
xticklabels={10,20,30,40},
y grid style={darkgray176},
ylabel={Average payoff per device},
ymin=0.52669776447995, ymax=0.73181197786969,
ytick style={color=black}
]
\addplot [very thick, blue, mark=*, mark size=2, mark options={solid}]
table {%
0 0.621090147742752
1 0.688500595638408
2 0.710793187557414
3 0.722488604533792
};
\addlegendentry{MAJP Solution}
\addplot [very thick, orange, mark=asterisk, mark size=2, mark options={solid}]
table {%
0 0.536021137815847
1 0.550443564214015
2 0.556123947735371
3 0.55928356513029
};
\addlegendentry{Non-cooperative}
\end{axis}

\end{tikzpicture}
    \caption{Performance comparison in terms of the average normalized payoff per device in the coalition.}
    \label{fig:comparision}
\end{figure}
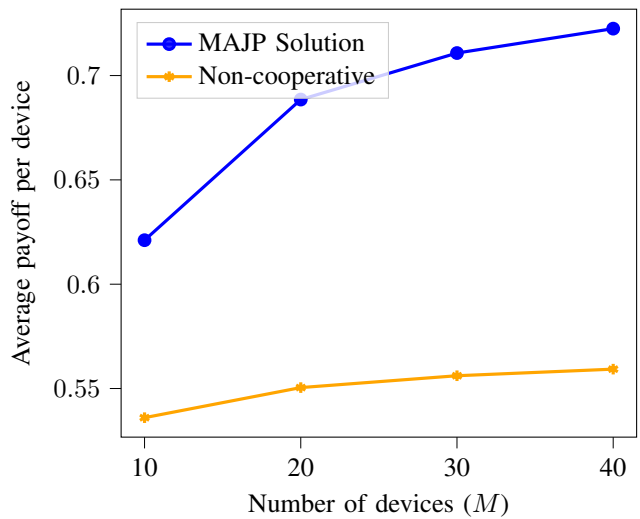
	\section{Conclusion and Discussion}\label{sec:conclusion}
	In this work, we have proposed an approach that establishes a trusted IoT data market by reinforcing collaboration opportunities between devices with correlated data to avoid information leakage. We set out to tackle the challenges posed due to the loop of mistrust in the data market; we jointly study three research questions (as indicated in Q1, Q2, and Q3), where we have shown devices with similar data types can cooperate in dealing with the price depression, data rivalry, and uncontrolled participation issues in the data market. We have formalized a network-wide optimization problem that maximizes the social value of coalition between the IoT devices of similar data types while minimizing the overall costs, defined in terms of network externalities, i.e., the impact of information leakage due to data correlation and the opportunity costs. The formulated problem is intractable due to binary constraint and is hard to solve directly given the presence of private information; thereby, we have developed a novel cooperative protocol, namely MAJP, that offered a sub-linear complexity in obtaining the solution using a preference-based coalition strategy. To that end, we have shown, via statistical analysis and numerical evaluations, our proposed approach provides benefits (around $32.72\%$ gain) as compared to the non-cooperative baseline, revealing truthful participation of devices without uncontrolled competition due to the information leakage and data rivalry.
	
	\rrv{We believe that our work and related works on data valuation/markets may have, over the long term, an impact on the way data privacy is regulated, as it generalises the current dominant paradigms of free vs. fully private data. The proposed approach could open additional benefits.} For example, knowing other's data properties a priori also indicates devices can learn when it is reasonable to collaborate for training learning models, as discussed in \cite{mpc}. \rv{An interesting direction for future work is to consider a more practical network setup with intermittent links and resource constraints for IoT devices. Another aspect is to better quantify the amount of privacy leakage by using notions of differential privacy \cite{dwork2014algorithmic} or multi-party computation \cite{mpc} and develop closer-to-the-real-world utility models. \rrv{Together with data privacy, security aspects during data trading, for instance, the provenance and authenticity of data, should be explored. We also foresee challenges in implementing the proposed solution in a practical scenario, given the scale of additional signalling and accuracy-complexity trade-offs required to identify the data properties.} As discussed, the merits of the proposed method are a stable, low-complexity, weak Pareto-optimal solution, which seems challenging to guarantee in an online setting. To that end, it would be interesting to study the scalability issues in a purely distributed network architecture.}
	\bibliographystyle{ieeetr}
	\bibliography{ref}
	\vskip -2.2\baselineskip plus -1fil
	\begin{IEEEbiography}[{\includegraphics[width=1in,height=1.25in,clip,keepaspectratio]{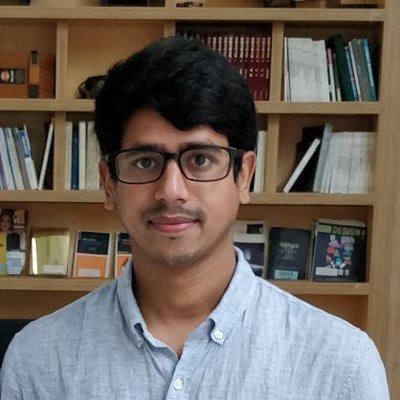}}]
		{\bf Shashi Raj Pandey} (Member, IEEE) received his  B.E. degree in Electrical and Electronics with a specialization in communication engineering from Kathmandu University, Nepal, and the Ph.D. degree in Computer Science and Engineering from Kyung Hee University, Seoul, South Korea. He is currently an Assistant Professor at the Department of Electronics Systems, Aalborg university. Prior, he was a Postdoctoral Researcher at the Connectivity Section, Aalborg University from 2021 to 2023. He served as a Network Engineer at Huawei Technologies Nepal Co. Pvt. Ltd, Nepal from 2013 to 2016.  His research interests include network economics, game theory, wireless communications, data markets and distributed machine learning. He was the receipt of the Best Paper Award at several conferences, including IEICE APNOMS 2019. He was a Member at Large at the IEEE Communication Society Young Professionals 2020 -- 2021. He currently serves as a  Member at Large in the IEEE Communication Society On-Line Content Board and is in the editorial advisory board of IEEE's Spectrum The Institute 2022 -- 2024. He is an affiliated member of the Pioneer Center for AI, Denmark.
	\end{IEEEbiography}
	\vskip -2.2\baselineskip plus -1fil
	\begin{IEEEbiography}[{\includegraphics[width=1.2in,height=1.25in,clip,keepaspectratio]{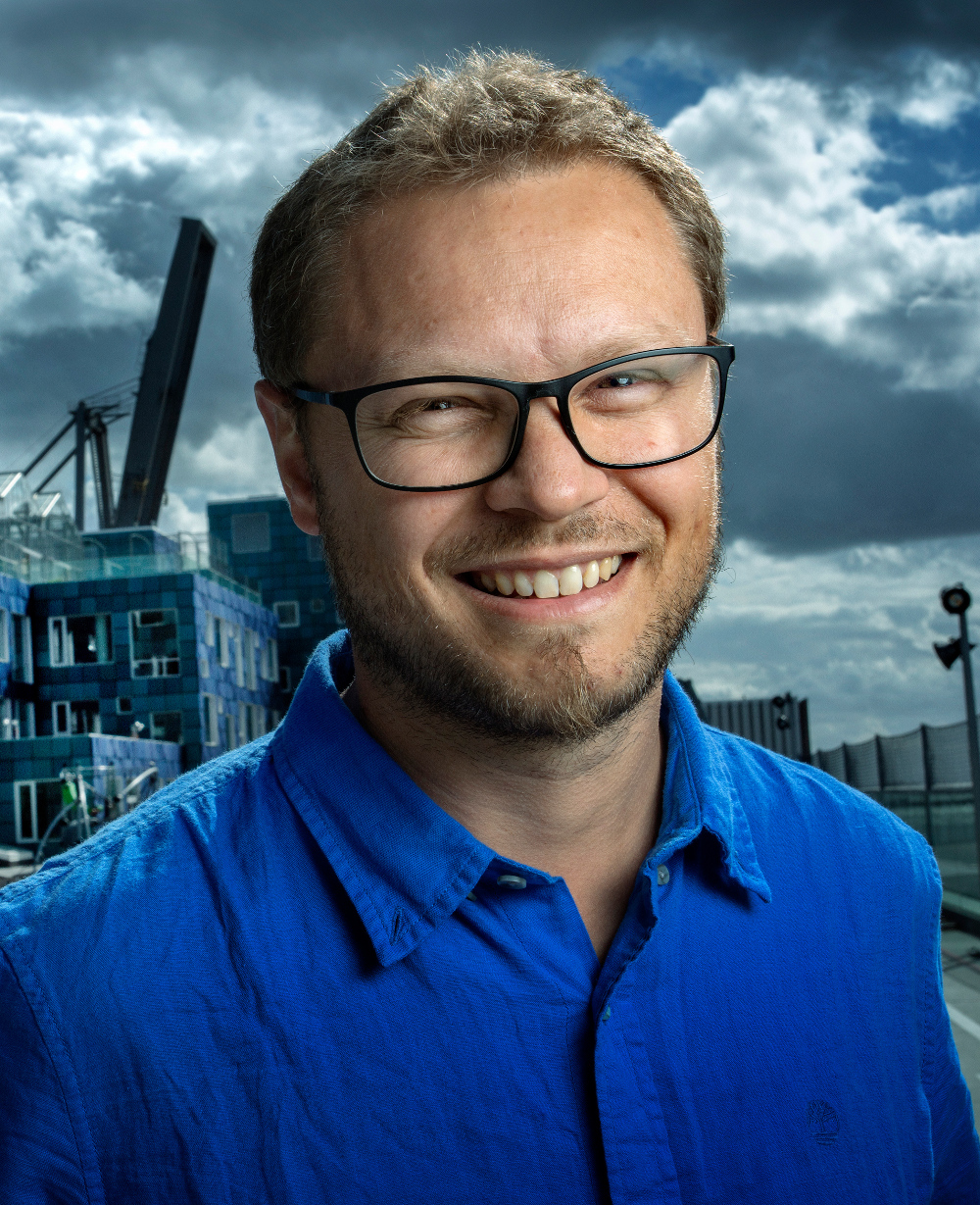}}]
		{\bf Pierre Pinson} (Fellow, IEEE) received the M.Sc. degree in applied mathematics from the National Institute of Applied Sciences (INSA), Toulouse, France, in 2002 and the Ph.D. degree in energetics from Ecole des Mines de Paris, France, in 2006. He is the chair of data-centric design engineering at Imperial College London, United Kingdom, Dyson School of Design Engineering. He is also an affiliated professor of operations research and analytics with the Technical University of Denmark and a chief scientist at Halfspace (Denmark). He is the editor-in-chief of the International Journal of Forecasting. His research interests include analytics, forecasting, optimization and game theory, with application to energy systems mostly, but also logistics, weather-driven industries and business analytics. He is a Fellow of the IEEE, an INFORMS member and a director of the International Institute of Forecasters (IIF).
	\end{IEEEbiography}
	\vskip -2.2\baselineskip plus -1fil
	\begin{IEEEbiography}[{\includegraphics[width=1in,height=1.25in,clip,keepaspectratio]{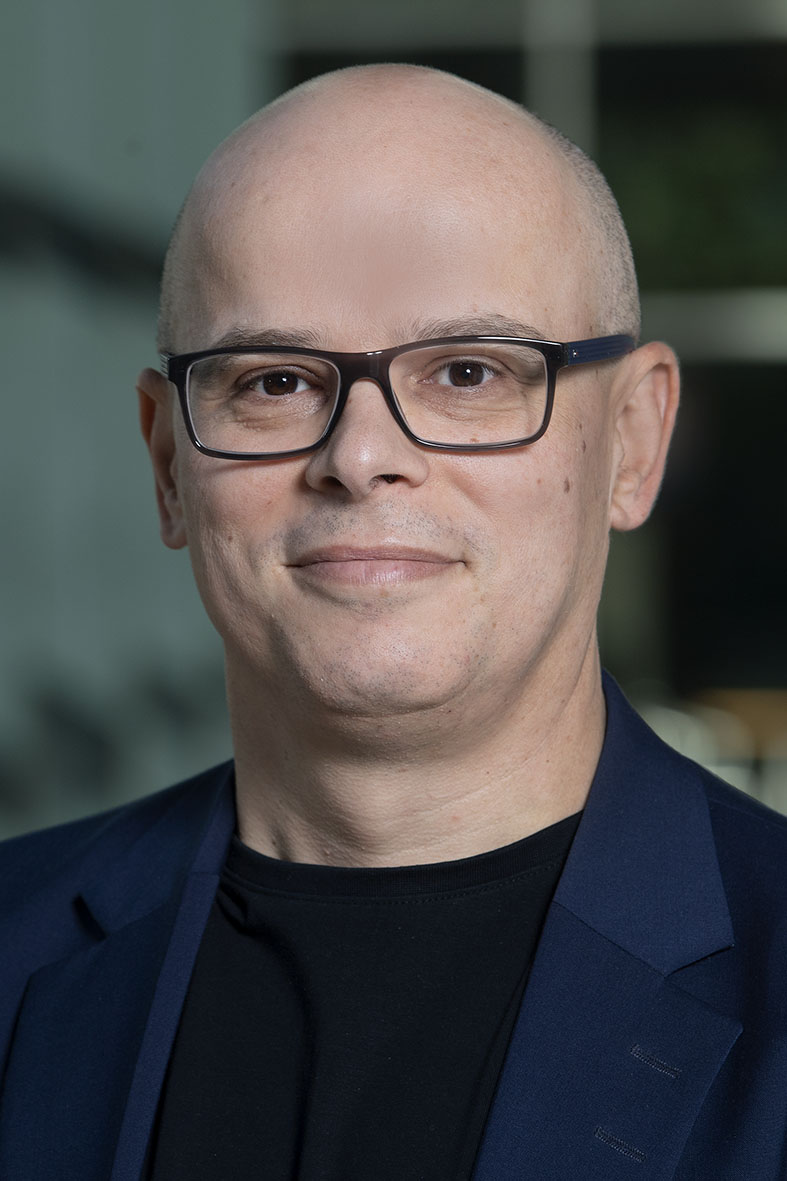}}]
		{\bf Petar Popovski} (Fellow, IEEE) is a Professor at Aalborg University, where he heads the section on Connectivity and a Visiting Excellence Chair at the University of Bremen. He received his Dipl.-Ing and M. Sc. degrees in communication engineering from the University of Sts. Cyril and Methodius in Skopje and the Ph.D. degree from Aalborg University in 2005. He received an ERC Consolidator Grant (2015), the Danish Elite Researcher award (2016), IEEE Fred W. Ellersick prize (2016), IEEE Stephen O. Rice prize (2018), Technical Achievement Award from the IEEE Technical Committee on Smart Grid Communications (2019), the Danish Telecommunication Prize (2020) and Villum Investigator Grant (2021). He was a Member at Large at the Board of Governors in IEEE Communication Society 2019-2021. He is currently an Editor-in-Chief of IEEEE JOURNAL ON SELECTED AREAS IN COMMUNICATIONS. He also serves as a Vice-Chair of the IEEE Communication Theory Technical Committee and the Steering Committee of IEEE TRANSACTIONS ON GREEN COMMUNICATIONS AND NETWORKING. Prof. Popovski was the General Chair for IEEE SmartGridComm 2018 and IEEE Communication Theory Workshop 2019. His research interests are in the area of wireless communication and communication theory. He authored the book ``Wireless Connectivity: An Intuitive and Fundamental Guide'', published by Wiley in 2020.
	\end{IEEEbiography}	
\end{document}